\documentclass{article}

\usepackage[margin=2cm]{geometry}
\usepackage{graphicx}
\usepackage{epstopdf}
\usepackage[caption=false]{subfig}
\usepackage{amsfonts}
\usepackage{amsthm}
\usepackage{amsmath}
\usepackage{amssymb}
\usepackage{url}
\usepackage[hidelinks,pdftex,
            pdfauthor={Noah Weninger and Ricardo Fukasawa},
            pdftitle={Interdiction of minimum spanning trees and other matroid bases},
            pdfsubject={Combinatorial Optimization},
            pdfkeywords={Interdiction, Bilevel, Stackelberg Game, Matroid, Minimum Spanning Tree, Algorithms, Most Vital Edges, Edge Blocker}
			]{hyperref}
\usepackage[vlined,linesnumbered,commentsnumbered]{algorithm2e}
\usepackage[capitalise]{cleveref}
\usepackage{mathtools}
\usepackage{xcolor}
\usepackage{todonotes}
\setuptodonotes{inline}  
\todostyle{response}{color=green}
\usepackage{booktabs}
\usepackage{pgfplotstable} \pgfplotsset{compat=1.18}

\usepackage[numbers,sort&compress]{natbib}
\bibpunct[, ]{[}{]}{,}{n}{,}{,}
\makeatletter
\def\NAT@def@citea{\def\@citea{\NAT@separator}}
\makeatother

\theoremstyle{plain}
\newtheorem{theorem}{Theorem}[section]
\newtheorem{lemma}[theorem]{Lemma}
\newtheorem{corollary}[theorem]{Corollary}
\newtheorem{proposition}[theorem]{Proposition}

\theoremstyle{definition}

\theoremstyle{remark}

\def\U{\mathcal{U}}
\def\L{\mathcal{L}}
\def\I{\mathcal{I}}
\def\B{\mathcal{B}}

\def\Z{\mathbb{Z}}

\def\rank{\operatorname{r}}
\def\lopt{F}
\def\bmin{\B^{\min}}

\def\OPT{\operatorname{OPT}}
\def\sm{\setminus}
\def\se{\subseteq}
\def\propbound{P_{\text{bound}}}
\def\propiinx{P_{\text{succ}}}

\def\propapprox{P_{\text{approx}}}

\SetKwProg{Fn}{function}{}{end}
\SetKwProg{Ds}{structure}{}{end}
\SetKwFunction{MinCut}{MinCut}
\SetKwFunction{SubsetSum}{SubsetSum}
\SetKwFunction{ComputeBound}{ComputeBound}
\SetKwFunction{BranchAndBound}{BranchAndBound}
\SetKwFunction{LowerBound}{LowerBound}
\SetKwFunction{UpperBound}{UpperBound}
\SetKwFunction{Interdict}{Interdict}
\SetKwFunction{UnInterdict}{UnInterdict}
\SetKwFunction{Skip}{Skip}
\SetKwFunction{UnSkip}{UnSkip}
\SetKwFunction{MST}{MST}
\SetKwFunction{DynInterdict}{DynamicInterdiction}
\SetKwFunction{Init}{Init}
\SetKwFunction{Union}{Union}
\SetKwFunction{Undo}{Undo}
\SetKwFunction{Find}{Find}
\SetKwFunction{Next}{Next}
\SetKwFunction{Prev}{Prev}
\SetKwFunction{Push}{Push}
\SetKwFunction{Pop}{Pop}
\SetKwFunction{Insert}{Insert}
\SetKwFunction{Erase}{Erase}
\SetKw{Break}{break}
\SetKw{Continue}{continue}
\SetKw{Output}{output}

\numberwithin{equation}{section}

\begin{document}

\title{Interdiction of minimum spanning trees and other matroid bases}

\author{
Noah Weninger and Ricardo Fukasawa\\
Department of Combinatorics \& Optimization, University of Waterloo, Canada\\
\texttt{\{nweninger,rfukasawa\}@uwaterloo.ca}
}

\maketitle

\begin{abstract}
In the minimum spanning tree (MST) interdiction problem,
we are given a graph $G=(V,E)$ with edge weights, and
want to find some $X\se E$ satisfying a knapsack constraint such
that the MST weight in $(V,E\sm X)$ is maximized.
Since MSTs of $G$ are the minimum weight bases in the graphic
matroid of $G$, this problem is a special case of {\em matroid interdiction} on a matroid $M=(E,\I)$,
in which the objective is instead to maximize the minimum weight of a basis of $M$ which is disjoint from $X$.
By reduction from 0-1 knapsack, matroid interdiction is NP-complete, even for uniform matroids.

We develop a new exact algorithm to solve the matroid interdiction problem. 
One of the key components of our algorithm is a dynamic programming upper bound which only requires that a simpler 
discrete derivative problem can be calculated/approximated for the given matroid.
Our exact algorithm then uses this bound within a custom branch-and-bound algorithm.

For different matroids, we show how this discrete derivative can be calculated/approximated. 
In particular, for partition matroids, this yields a pseudopolynomial time algorithm. 
For graphic matroids, an approximation can be obtained by solving a sequence of minimum cut problems, which we apply to the MST interdiction problem.
The running time of our algorithm is asymptotically faster than the best known
MST interdiction algorithm, up to polylog factors.
Furthermore, our algorithm achieves state-of-the-art computational performance:
we solved all available instances from the literature, and in many cases
reduced the best running time from hours to seconds.
\end{abstract}

\section{Introduction}
\label{sec:intro}
In an {\em interdiction problem}, we seek to solve some combinatorial optimization problem on
a discrete structure, but there is an adversary that can interdict (i.e., destroy)
some parts of the structure before the optimization problem is solved \cite{smith2020survey}.
The adversary aims to make the optimal solution value of the optimization problem as bad as possible.
Typically, for each part of the structure that is interdicted, the adversary has to pay a price,
and they have a limited budget (i.e., they are restricted by a knapsack constraint).

Interdiction has been studied since the 1950s, and was originally motivated by military applications,
but is now seen to be widely applicable.
For instance, the minimum $s$-$t$ cut problem can be seen
as the problem of interdicting $s$-$t$ connectivity in a graph by deleting edges;
this now-ubiquitous problem originated in a 1955 US Air Force secret report with the objective
of disrupting the Soviet railway system \cite{schrijver2002history}.
Another commonly studied interdiction problem is shortest $s$-$t$ path interdiction,
in which the adversary tries to increase the shortest $s$-$t$ path length in a graph as much as possible
by deleting edges \cite{israeli2002shortest}. Although the discrete structure is often a graph, other examples
exist, such as knapsack interdiction \cite{weninger2023fast}, in which the adversary deletes items from a knapsack,
or matroid interdiction, the focus of this paper.

In recent years, interdiction problems have typically been interpreted as 
{\em bilevel integer programming} (BIP) problems, which can be viewed as a generalization of
integer programming (IP) to two-round two-player games. In a BIP,
there are two IPs, called the upper level and lower level, between which some variables are shared.
An optimal solution to a BIP problem can be seen as the optimal first-player strategy in a two player game.
In this game, the first player, called the leader, chooses a point $X$ in the upper level feasible region.
The second player, called the follower, then picks an optimal solution $Y$ for the lower level IP,
for which feasibility and optimality may depend on the choice of $X$.
In the most general form, the feasibility of $X$ can also depend on the choice of $Y$.
The game then ends and the objective function
is calculated as a linear function of $X$ and $Y$. The goal is for the leader to pick an $X$ which
maximizes (or minimizes) this objective function. The terms leader and follower
are borrowed from the theory of Stackelberg games \cite{von1952theory}.

The advantage of viewing interdiction problems as BIPs is that a significant
amount of work has been written on general purpose algorithms for solving BIPs \cite{kleinert2021survey}.
However, a key challenge in this approach is that most BIPs are $\Sigma_2^p$-hard and
hence cannot be modelled as IPs
without using an exponential number of constraints and/or variables.
So, typically these algorithms reduce BIP to the problem of solving a series of IPs---exponentially
many in the worst case.
However, while these general purpose methods can solve small problem
instances of most interdiction problems, current solvers are typically very poor
at identifying and exploiting the (abundant) structure present in interdiction
problems. As such, performance improvements of many orders of magnitude are
typically seen when switching from general purpose BIP algorithms to
problem-specific interdiction algorithms. This is in contrast to many
classical combinatorial optimization problems,
where general purpose IP solvers are often very competitive.

Matroid interdiction is only NP-hard (as opposed to $\Sigma_2^p$-hard),
but there is a lack of algorithms for solving it which are fast in practice.
The objective of this work is to further progress our understanding of the
structure present in interdiction problems by the development of problem-specific
algorithms for matroid interdiction. We hope that these insights can later be integrated
into general purpose BIP solvers.

\subsection{Problem statement}
\label{sec:problemstatement}

We assume the reader already has basic familiarity with matroids; the below definitions serve
primarily to standardize notation. For a relevant introduction to matroids,
see Chapter 8 in \cite{cunningham1998combinatorial}.
A matroid $M=(E,\I)$ consists of a set $E$, called the {\em ground set}, and
$\I$, a family of subsets of $E$, called the {\em independent sets}.
If $J\in\I$, we say $J$ is {\em independent}.
For $M$ to be a matroid, we must have:
\begin{enumerate}
	\item $\emptyset\in\I$,
	\item for all $A\in\I$ and $B\se A$, $B\in\I$ (i.e., $\I$ is {\em downwards closed}), and
	\item if $A$ and $B$ are maximal independent subsets of a set $J\se E$, then $|A|=|B|$.
\end{enumerate}
We call maximal sets in $\I$ {\em bases}.
We work in the oracle model: we assume we have a function
which, given a set $J\se E$, answers whether $J\in\I$ in a single unit of time.
Due to these properties, a minimum (respectively maximum) weight basis can
be found using a greedy algorithm which adds elements in order of
non-decreasing (non-increasing) weight whenever doing so retains independence \cite{cunningham1998combinatorial}.
For any $A\se E$, we define the {\em rank} of $A$ to be $\rank(A):=\max\{|B|:B\se A,B\in\I\}$.
Note that $B$ is a basis if and only if $\rank(B)=\rank(E)$.

In the paper we refer explicitly to three types of matroids: uniform, partition, and graphic.
The independent sets $\I$ of a {\em uniform} matroid on a ground set $E$ are defined by
$\I=\{S\se E:|S|\le k\}$ for some given $k$.
Partition matroids generalize uniform matroids: they are defined by a partition
$E_1\cup E_2\cup\dots\cup E_{\ell}=E$ and integers $k_1,k_2,\dots,k_{\ell}$;
the independent sets are $\I=\{S\se E:|S\cap E_i|\le k_i~\forall i\in\{1,\dots,\ell\}\}$.
The {\em direct sum} of two matroids $M_1=(E_1,\I_1)$ and $M_2=(E_2,\I_2)$
where $E_1\cap E_2=\emptyset$
is a matroid $M=(E,\I)$ where $E=E_1\cup E_2$ and
$\I=\{S:S\cap E_1\in\I_1\text{ and }S\cap E_2\in\I_2\}$.
All uniform matroids are partition matroids, but unlike uniform matroids,
the class of partition matroids is closed under direct sum.
Given an undirected graph $G=(V,E)$, the {\em graphic} matroid $M=(E,\I)$ defined on $G$ has
$\I=\{S\se E:S\text{ is acyclic}\}$. Throughout the paper, we assume that $G$ is connected.
 So, the bases of the graphic matroid
are the maximal acyclic subsets of edges of $G$, i.e., they are spanning trees.

Let $M=(E,\I)$ be a matroid and let $m=|E|$.
We assume throughout the paper that $E=\{1,2,\dots,m\}$.
An instance of the matroid interdiction problem is defined by $M$,
a weight vector $w\in\Z^m$, a cost vector $c\in\Z^m_{\ge0}$,
and a capacity $C\in\Z_{\ge0}$.
Let $\B=\{A\se E:\rank(A)=\rank(E)\}$ be the set of bases of $M$.
For convenience, we write $x(S)=\sum_{i\in S}x_i$ for a vector $x$ and set $S$.
The matroid interdiction problem can then be stated as follows.
\begin{gather}
	\OPT=\max_{X\in\U}\,\min_{Y\in \L(X)}\,w(Y) \label{eq:matint} \tag{MI} \\
	\begin{aligned}
		\text{where }\U&=\left\{X\se E:c(X)\le C\right\}, \;\; \mbox{(upper level)}\\
		\text{and }\L(X)&=\left\{Y\se E\sm X:Y\in\B\right\}.\;\;\mbox{(lower level)}
	\end{aligned} \nonumber
\end{gather}
We often say that the upper-level variables $X$ are selected by a decision maker
called the {\em leader}, and the lower-level variables $Y$ are selected by another
decision maker called the {\em follower}.
We call this the {\em max-min variant} of the problem; the {\em min-max variant} swaps both max for min
and min for max in the objective function.
If we are concerned with exact solutions, the
choice of working with the $\max$-$\min$ variant is arbitrary
because negating the weights switches between the two variants:
\begin{equation}
\arg\max\{\min\{w(Y):y\in\L(X)\}:X\in\U\}=-\arg\min\{\max\{-w(Y):Y\in\L(X)\}:X\in\U\}.
\end{equation}
It is easily seen that we may assume $w\ge1$ by adding a sufficiently
large constant to all weights, and we may assume $c\ge1$ because any element $e$ with
$c_e=0$ can be deleted without affecting the optimal solution.

As is typical in bilevel programming, we say that a solution $(X,Y)$ is {\em feasible}
if $X\in\U$ and $Y\in\arg\min\{w(Y):Y\in\L(X)\}$, that is, to have bilevel feasibility
$Y$ must be optimal for the lower-level problem.
Since the minimum weight basis of a matroid can be found in polynomial time by the greedy algorithm, feasibility
can be verified in polynomial time, and hence the matroid interdiction decision problem is in NP.

We adopt some standard terminology from matroid theory: a matroid which cannot be written
as a direct sum of two nonempty matroids is called {\em connected}, and a maximal connected
submatroid of a matroid is called a {\em component} of that matroid.
So, if $C_1$ and $C_2$ are components of a matroid $M$, $I_1$ is independent in $C_1$ and
$I_2$ is independent in $C_2$, then $I_1\cup I_2$ is independent in $M$.
Throughout the paper, we assume that
\begin{equation}
\text{if $i$ and $j$ are elements in the same component of a matroid
and $i<j$, then $w_i\le w_j$.} \tag{A1}\label{ass:order}
\end{equation}
While there may be multiple minimum weight bases of a matroid,
given that $E=\{1,\dots,m\}$, there is always a unique {\em lexicographically smallest basis},
and if \labelcref{ass:order} is satisfied, then by the greedy algorithm the lexicographically smallest basis is of minimum weight.
We assume that the ordering of matroid elements satisfies \labelcref{ass:order} throughout the paper.
Given a matroid  $M=(E,\I)$, let $\bmin_{\!M}(X)$
be the lexicographically smallest basis $B\se E\sm X$ of $M$.
This function is only defined for those $X$ where there exists a basis $B\se E\sm X$.
We may omit the subscript $M$ when it is implied by the context.

Given the lexicographically smallest basis $B$ of a matroid $M$
and an element $e\in B$,  we say that the {\em replacement element} for $e$
(if one exists) is the element $e'$ such that $\{e'\}=B'\sm B$ where $B'$ is the lexicographically smallest
basis of $M\sm \{e\}$. Note that the replacement element can be found by the greedy algorithm.

A closely related problem to matroid interdiction is the {\em minimum cost matroid blocker problem},
which can be formulated as follows.
\begin{align*}
	&\min_{X\se E} c(X)\\
	&\text{s.t. }\left[\min_{Y\in\L(X)} w(Y)\right]\ge R
\end{align*}
Here, $R$ is a given constant which can be though of as the {\em target weight}.
It is easy to see that these two problems are polynomial-time equivalent by performing
binary search on $C$ or $R$. We focus on the matroid interdiction problem for the majority of the paper,
but in \cref{sec:mincostblocker} we show that our algorithm can easily be adapted to directly solve
the minimum cost blocker problem without incurring the run time overhead of binary search.

\subsection{Prior work}

The simplest and earliest-studied special case of matroid interdiction
is the {\em most vital MST edge} problem, which was introduced by Hsu et al
in 1991 \cite{hsu1991finding}. In this variant we wish to find a single edge $e$
which maximizes the MST weight of $G\sm\{e\}$.
This is the special case of matroid
interdiction where $M$ is graphic, $c_e=1$ for all $e\in E$, and $C=1$. A long line
of work led to a recent algorithm which solves this problem in
time $O(m\log m+n)$, or $O(m+n)$ if the edges are assumed to already be sorted
by weight \cite{bader2019simple}. A few works have studied variants of the most vital matroid element problem:
one paper performs sensitivity analysis on minimum weight matroid bases \cite{libura1991sensitivity},
another studies the variant where the leader can
increase the weights of the matroid elements by paying a cost per unit of weight increased
\cite{frederickson1998algorithms},
and a very recent paper concerns parametric matroid interdiction \cite{hausbrandt2024parametric}.

A natural generalization of this is the $k$-most vital MST edges problem, in which $M$ is
graphic, $c_e=1$ for all $e\in E$ and $C=k$ for some constant $k$.
NP-hardness of the $k$-most vital edges
problem was shown in \cite{frederickson1999increasing} by reduction from minimum $k$-cut.
Later, this hardness result was extended to show NP-hardness for the variant of the problem where vertices
are interdicted instead of edges \cite{bazgan2013critical}.
The same paper also includes various approximation algorithms and inapproximability results.
Since the problem is only NP-hard when $k$ is part of the input, an interest developed in algorithms
which run in polynomial time for fixed $k$ \cite{shen1999finding}.
After a series of papers, a running time of $O(n^k\log\alpha((k-1)(n+1),n))$ was achieved
\cite{liang2001finding,bazgan2012efficient}.
Most recently, new mixed integer programming formulations and branch-and-bound algorithms
were proposed for the $k$-most vital MST edges problem
which have good performance in computational tests for small $k$
(i.e., $k\le 5$ on instances with hundreds of vertices)
\cite{bazgan2012efficient}. We compare our results to this paper in \cref{sec:comp}.

The first paper to study MST interdiction (then called the
{\em most vital edges in the MST} problem) in full generality presents a simple
branch-and-bound algorithm \cite{lin1993most}.
To the best of our knowledge, this algorithm was never evaluated computationally.
A later paper also established an $O(\log n)$-approximation
algorithm \cite{frederickson1999increasing}.
This was subsequently improved to a 14-approximation \cite{zenklusen20151},
and eventually a $4$-approximation \cite{linhares2017improved}. The paper which shows
a $4$-approximation also proves the surprising result that the
{\em maximum} spanning tree interdiction problem cannot be approximated to any constant factor 
under the small set expansion hypothesis.
However, a 2-pseudo-approximation (a polytime algorithm that either returns a 2-approximate
solution or a solution that exceeds the budget by a factor of 2) is known for min-max matroid interdiction, and hence for maximum spanning tree interdiction \cite{chestnut2017interdicting}.
In fact, this pseudo-approximation even works for more general case of interdicting the
maximum independent set weight, rather than just the maximum basis weight, as well as for a
wide variety of other interdiction problems.

In 2024, Luis Salazar-Zendeja published their PhD thesis on MST interdiction, which
proposes and computationally evaluates a variety of MIP-based approaches \cite{salazar2022models}.
The computational results are not compared to prior works, but the largest instances considered
(in terms of number of vertices/edges and capacity) are smaller than those
studied in previous papers, and the running times are of the same order of magnitude.
Important contributions of the thesis include the application of classical but previously unexplored
techniques including branch-and-price and Bender's decomposition,
and the generalization of all results to the partial interdiction variant of the problem, in which the
leader can only slightly increase the weight of MST edges rather than block them entirely.

The MST edge blocker problem is the focus of \cite{wei2021integer},
which studies various integer programming formulations for the problem and evaluates
their performance computationally.
We compare our results to this paper in \cref{sec:comp}.
A follow-up paper generalizes the theoretical aspects of this work to greedoids,
among other problems \cite{wei2022integer}.

A number of closely related problems have been studied.
One is the {\em rank reduction} problem,
where we wish to find a minimum cardinality subset of the elements of a matroid
whose deletion reduces the rank by at least $k$ \cite{joret2015reducing}.
Although many parallels can be drawn between the known hardness results for rank reduction
and matroid interdiction, we are not aware of any work which connects them explicitly.
Another recent paper studies the interdiction problem where both the upper and lower
level feasible regions are given by the independent sets of partition matroids over
a common ground set \cite{ketkov2024class}.
Finally, our recent work~\cite{weninger2023fast} on knapsack interdiction also uses branch-and-bound with
bounds computed via dynamic programming, and achieves state-of-the-art computational
results. The success of those techniques motivated
our approach in this work.

\subsection{Our contributions}

In \cref{sec:ub}, we introduce a new framework for computing upper bounds for 
interdiction problems. We apply this framework to exactly solve the partition
matroid interdiction problem, and derive an upper bound for MST interdiction.
In \cref{sec:exact} we present an exact branch-and-bound algorithm for solving matroid interdiction
utilizing the bounds derived in the previous section.
We prove that even without the upper bounds, this algorithm matches (up to polylog factors)
or improves upon the best known
asymptotic running time of an enumerative algorithm for MST interdiction with unit costs.
Furthermore, our algorithm is more general as it applies to any matroid interdiction problem.
We also introduce a strong heuristic lower bound to initialize the branch-and-bound search.
In \cref{sec:comp} we compare an implementation of our algorithm in computational tests
to previously published algorithms. We find that our algorithm is significantly faster
than previous approaches, often by a few orders of magnitude.
We also generate some new problem instances, which we use to investigate what qualities make an instance
hard to solve, and how the various features of our solver contribute to its performance.
Finally, in \cref{sec:dual} we investigate the problem variant in which the leader
can force elements to be included in every basis, rather than excluded.

\section{Upper bounds}
\label{sec:ub}

\def\bound{\delta}
\def\update{\pi}
\def\state{\phi}

In this section, we describe a general framework for upper bounding
interdiction problems with knapsack-like upper level feasibility sets.
The key idea is to reduce the interdiction problem to a type of knapsack problem
where the item profits bound the change in the lower level
objective function when a given element is interdicted. We demonstrate that
for many matroids, when the lower level problem is to find the minimum basis weight,
this discrete-derivative-like problem is easier to solve
than the full interdiction problem and yields good bounds or even
exact solutions when integrated with dynamic programming (DP).
This way of computing upper bounds
also is ideally suited for use in a branch-and-bound scheme,
because the bound can be precomputed for every branch-and-bound node and
accessed in $O(1)$ time. Our exact branch-and-bound algorithm for matroid
interdiction which uses these bounds is described in \cref{sec:exact}.

We assume for now that the set of feasible interdictions is $\U=\{X\se E:c(X)\le C\}$,
as described in Section~\ref{sec:intro}, although later we argue that this can be relaxed somewhat.
From this point until the start of \cref{sec:uniform},
the only assumption that we make about $\L(X)$ is that if $Y\in\L(X)$, then $Y\se E\sm X$.
Let $F(X)=\min\{w(Y):Y\in\L(X)\}$ be the optimal
solution to the follower's problem induced by a given interdiction set $X\in\U$,
or $\infty$ if the follower's problem is infeasible for the given $X$.
The value we desire to upper bound is $\OPT=\max\{F(X):X\in\U\}$.

For a set $X\se\Z$, let $X_{\le k}$ be the set $\{j\in X:j\le k\}$,
and similarly for $X_{\ge k}$, $X_{<k}$, etc.
Given any $\hat X\in\U$ and $i\in E$, we can express $\max\{\lopt(X):X\in\U,X_{<i}=\hat X\}$
as $\lopt(\hat X)+\max\{\lopt(X)-\lopt(\hat X):X\in\U,X_{<i}=\hat X\}$.
In this section, we show that this view offers a clear advantage: upper bounds on
$\max\{\lopt(X)-\lopt(\hat X):X\in\U,X_{<i}=\hat X\}$ can be computed for all $i\in E$ and
$\hat X\in\U$ as the entries of a DP table.
Moreover, this structure will be exploited in our branch-and-bound algorithm (presented in \cref{sec:exact}), in which
branch-and-bound nodes are identified by $\hat X$ and $i$.

Before we present our DP algorithm, we comment that we formalize only the essential
aspects that are needed for such description, for the sake of generality. In particular, we choose to be
informal regarding some of the desired aspects, using terms like `small enough' or
`large enough', since these are guiding principles that need to be made concrete depending on the specific lower level problem.

The dynamic program will be over a set of possible states $T$ and requires mapping
each $X\in\U$ to some {\em state} $s\in T$ by a function $\state(X)$.
The goal is for $|T|$ to be small enough that we can use $s\in T$
to index the DP table (without using too much memory), but large enough that
$\state(\hat X)$ captures enough information about some $\hat X\in\U$ to be
able to compute a good bound for $\max\{\lopt(X)-\lopt(\hat X):X\in\U:X_{<i}=\hat X\}$.
We also define a function $\update(i,s)$ which describes how a state $s\in T$
transitions to a new state if element $i$ is interdicted
when in state $s$. We formalize this property below as $\propiinx$.

Finally, we need a function $\bound(i,r,s)$ which upper bounds how much the follower's objective
can increase by if element $i$ is interdicted with $r$ knapsack capacity 
remaining when in state $s$.
We formalize this property below as $\propbound$.
The assumption that this function can be computed faster
than solving the entire interdiction problem is key to our approach.
 We devote most of the second half of this section to
defining a suitable $\bound$ for various matroids.

Formally, these functions have the following signatures:
\begin{align*}
&\state:\U\to T\\
&\update:E \times T\to T\\
&\bound:E \times\{0,1,\dots,C\}\times T\to\Z
\end{align*}
These will be defined precisely later, depending on the specifics of the follower's problem.
We now formally define the two conditions that these functions should satisfy to be able
to apply our results.
In the below definitions, $\U$ and $T$ are implicit based on
the definitions of the given functions.
Condition $\propbound(\bound,\state,\update)$ is satisfied if and only if
for all $i\in E$, $r\in\{0,1,\dots,C\}$ and $s\in T$,
\begin{align*}
c_i\le r\implies\bound(i,r,s)\ge\max\begin{Bmatrix*}[l]
	\lopt(X_{<i}\cup\{i\})-\lopt(X_{<i})\\
	\text{s.t. }X\in\U,\,
	c(X_{<i})=C-r,\,
	\state(X_{<i})=s
\end{Bmatrix*}.
\end{align*}
This condition is the most important since it ultimately defines what the upper bound is.
The other conditions and functions are defined to make sure that DP transitions
from one state to another are consistent and can be computed for all possible states.
The next condition, $\propiinx(\state,\update)$, is satisfied if and only if for all $X$ such that $X\cup\{i\}\in\U$,
\begin{align*}
\state(X_{<i}\cup\{i\})=\update(i,\state(X_{<i})).
\end{align*}
This condition intuitively says that transitioning from the state $\phi(X_{<i})$ by
interdicting $i$ should lead us to the state $\phi(X_{<i}\cup\{i\})$.
We now define a dynamic program $f:\{1,\dots,m+1\}\times\{0,\dots,C\}\times T\to\Z$.
\begin{align*}
f(i,r,s)=\begin{cases}
0&\text{if }i>m,\\
f(i+1,r,s)&\text{if }i\le m\text{ and }c_i>r,\\
\max\begin{Bmatrix*}[l]
	f(i+1,r,s),\\
	f(i+1,r-c_i,\update(i,s))+\bound(i,r,s)
\end{Bmatrix*}&\text{otherwise.}
\end{cases}
\end{align*}
Note that this is a generalization of the standard Bellman recursion for 0-1 knapsack
(e.g., see Section 2.3 of \cite{pisinger1998knapsack});
the only difference is that we keep track of additional state $s$, which is
managed by the update function $\update$,
and that the item profit values $\bound(i,r,s)$ are a function of both $i$, $r$, and $s$, as opposed to only $i$.
For example, suppose we wish to solve a knapsack problem with costs $c_i$, capacity $C$ and profits $p_i$.
Then we can simply define $T=\{\epsilon\}$, $\state(X)=\epsilon$, $\update(i,\epsilon)=\epsilon$, and $\delta(i,r,\epsilon)=p_i$
(where $\epsilon$ is just a null element used to ensure $T$ is nonempty);
then $f(1,C,\epsilon)$ will solve the knapsack problem.

The following theorem formalizes the above intuition using this definition.
It follows immediately that $\OPT\le f(1,C,\state(\emptyset))+\lopt(\emptyset)$
by applying the theorem with $i=1$ and $\hat X=\emptyset$.
To simplify notation in the next two proofs we define $\Phi(X)=(C-c(X),\,\state(X))$
and $\Pi(i,(r,s))=(r-c_i,\pi(i,s))$.
The intuition for this is that the remaining knapsack capacity $r=C-c(X)$ and the state $s=\phi(X)$ can both be treated as
parts of a larger state $\Phi(X)$ with transition function $\Pi$.
Then, it is easy to see that
for all $X$ such that $X\cup\{i\}\in\U$,
\begin{align*}
	\Phi(X_{<i}\cup\{i\})=\Pi(i,\Phi(X_{<i})).
\end{align*}
Some slight abuse of notation aside, this is equivalent to saying that $\propiinx(\Phi,\Pi)$ holds,
assuming that $\propiinx(\state,\update)$ does.
By doing this, we do not need to individually reason about both $r$ and $s$ in the proof, and can instead treat
them as a unit $(r,s)$.
To simplify notation we also interchangeably use $f(i,r,s)$ and $f(i,(r,s))$, and similarly for $\bound$ and $\pi$.
\begin{theorem}
\label{thm:ubCorrect}
Given an upper-level feasible set $\,\U$, state set $T$, and functions $\bound$, $\state$, and $\update$ satisfying
$\propbound(\bound,\state,\update)$ and
$\propiinx(\state,\update)$, we have that for all $i\in E$ and $\hat X\in\U$,
\[f(i,\Phi(\hat X_{<i}))\ge\max\{\lopt(X)-\lopt(X_{<i}):X\in\U,\,X_{<i}=\hat X_{<i}\}.\]
\end{theorem}
\begin{proof}
We prove by induction on $i$ from $m+1$ to $1$ that for all $r\in\{0,1,\dots,C\}$ and $s\in T$,
\begin{align}
	f(i,r,s)&\ge\max\begin{Bmatrix*}[l]
		\lopt(X)-\lopt(X_{<i}):X\in\U,\,\Phi(X_{<i})=(r,s)
	\end{Bmatrix*}.
	\label{eqn:indHyp}
\end{align}
Since $X_{<i}=\hat X_{<i}$ implies $\Phi(X_{<i})=\Phi(\hat X_{<i})$,
this proves the theorem statement by taking $(r,s)=\Phi(\hat X_{<i})$:
\begin{align*}
	\labelcref{eqn:indHyp}&\ge\max\begin{Bmatrix*}[l]
		\lopt(X)-\lopt(X_{<i})
		:X\in\U,\,
		X_{<i}=\hat X_{<i}
	\end{Bmatrix*}.
\end{align*}
For the base case of the induction, suppose that $i=m+1$.
Since $E\sm X=E\sm X_{<m+1}$ for any $X\se E$,
we have $\lopt(X)=\lopt(X_{<i})$ and hence
$f(i,r,s)=0$ satisfies \cref{eqn:indHyp}.
Now assume $1\le i\le m$ and that \cref{eqn:indHyp} holds for
$i+1$ with arbitrary $r\in\{0,\dots,C\}$ and $s\in T$.
First suppose that $c_i>r$.
Then, applying the induction hypothesis (\cref{eqn:indHyp}),
\begin{align*}
	f(i,r,s)&=f(i+1,r,s)
	\ge \max\begin{Bmatrix*}[l]
		\lopt(X)-\lopt(X_{<i+1})
		:X\in\U,\,
		\Phi(X_{<i+1})=(r,s)
	\end{Bmatrix*}
\end{align*}
Since $c_i>r$ we have that $i\notin X$ for any $X\in\U$ with $C-c(X_{<i})=r$.
Therefore, for any such $X$, $X_{<i+1}=X_{<i}$, so we have $\Phi(X_{<i+1})=(r,s)$.
Hence, the result follows.
Now assume $c_i\le r$.
We claim that
\begin{align}
f(i+1,\Pi(i,r,s))+\bound(i,r,s)
\ge\max\begin{Bmatrix*}[l]
		\lopt(X)-\lopt(X_{<i}))
		:i\in X\in\U,\,
		\Phi(X_{<i})=(r,s)
	\end{Bmatrix*}.\label{eqn:claimInterdictCase}
\end{align}
To see this, first observe that by the induction hypothesis and $\propbound(\bound,\state,\update)$,
\begin{align}
	&f(i+1,\Pi(i,r,s))+\bound(i,r,s)\nonumber\\
	\ge&\max\begin{Bmatrix*}[l]
			\lopt(X)-\lopt(X_{<i+1})
			:X\in\U,\,
			\Phi(X_{<i+1})=\Pi(i,r,s)
		\end{Bmatrix*}\label{eqn:ub1}\\
	+&\max\begin{Bmatrix*}[l]
		\lopt(X_{<i}\cup\{i\})-\lopt(X_{<i})
		:X\in\U,\,
		\Phi(X_{<i})=(r,s)
	\end{Bmatrix*}\label{eqn:ub2}.
\end{align}
Now, by taking the intersection of the feasible sets $X$ in \cref{eqn:ub1,eqn:ub2}
and adding the objectives, we get that
\begin{align}
\labelcref{eqn:ub1}+\labelcref{eqn:ub2}\ge&\max\begin{Bmatrix*}[l]
		\lopt(X)-\lopt(X_{<i+1})
		+\lopt(X_{<i}\cup\{i\})-\lopt(X_{<i})\\
		\text{s.t. }X\in\U,\,
		\Phi(X_{<i+1})=\Pi(i,r,s),\,
		\Phi(X_{<i})=(r,s)
	\end{Bmatrix*}\label{eqn:ub3}.
\end{align}
If $X\in\U$ and $i\in X$, then since $\propiinx(\Phi,\Pi)$ holds we have $\Phi(X_{<i+1})=\Pi(i,\Phi(X_{<i}))$.
Hence,
\begin{align}
	\labelcref{eqn:ub3}\ge&\max\begin{Bmatrix*}[l]
		\lopt(X)-\lopt(X_{<i})
		:i\in X\in\U,
		\Phi(X_{<i})=(r,s)
	\end{Bmatrix*}.
\end{align}
This proves the claim.
We now complete the inductive step. The following is immediate from
the induction hypothesis and \cref{eqn:claimInterdictCase}.
\begin{align}
f(i,r,s)&=\max\{f(i+1,r,s),f(i+1,\Pi(i,r,s))+\bound(i,r,s)\}\nonumber\\
&\ge\max\begin{Bmatrix*}[l]
	\max\begin{Bmatrix*}[l]
		\lopt(X)-\lopt(X_{<i+1})
		:X\in\U,\,
		\Phi(X_{<i+1})=(r,s)
	\end{Bmatrix*},\\
	\max\begin{Bmatrix*}[l]
		\lopt(X)-\lopt(X_{<i})
		:i\in X\in\U,\,
		\Phi(X_{<i})=(r,s)
	\end{Bmatrix*}
\end{Bmatrix*}\label{eqn:ub5}.
\end{align}
Now, by restricting $X$ in the first inner max of \cref{eqn:ub5} to exclude $i$, we get
\begin{align}
\labelcref{eqn:ub5}
&\ge\max\begin{Bmatrix*}[l]
	\max\begin{Bmatrix*}[l]
		\lopt(X)-\lopt(X_{<i})
		:i\notin X\in\U,\,
		\Phi(X_{<i})=(r,s)
	\end{Bmatrix*},\\
	\max\begin{Bmatrix*}[l]
		\lopt(X)-\lopt(X_{<i})
		:i\in X\in\U,\,
		\Phi(X_{<i})=(r,s)
	\end{Bmatrix*}
\end{Bmatrix*}.\label{eqn:ub6}
\end{align}
Finally, we can merge the two inner max terms into a single term by noticing that
the first considers all $X$ with $i\notin X$ and the second considers all $X$
with $i\in X$, but they are otherwise identical.
\begin{align*}
\labelcref{eqn:ub6}
&=\max\begin{Bmatrix*}[l]
	\lopt(X)-\lopt(X_{<i})
	:X\in\U,\,
	\Phi(X_{<i})=(r,s)
\end{Bmatrix*}.
\end{align*}
This proves \cref{eqn:indHyp}, as desired.
\end{proof}
We will use this theorem in \cref{sec:exact} to compute upper bounds for our branch-and-bound scheme.
A natural question to ask is how loose/tight this upper bound can be.
The next theorem gives some answer to that question by
showing that if we have additive approximation bound of $\alpha$ for $\bound$,
there we have an additive approximation bound of $m\alpha$ for $f$.
We formalize the notion of having an additive approximation bound of $\alpha$ for $\bound$ with
the condition $\propapprox(\bound,\state,\update,\alpha)$, which is satisfied if and only if
for all $X\in\U$ and $i\in E$,
\begin{align*}
&c_i\le C-c(X_{<i})\implies\bound(i,C-c(X_{<i}),\state(X_{<i})))
\le\alpha+\lopt( X_{<i}\cup\{i\})-\lopt( X_{<i}).
\end{align*}
\begin{theorem}
\label{thm:approx}
Suppose we have an upper-level feasible set $\,\U$, state set $T$, and functions $\bound$, $\state$, and $\update$.
Assume that $\propiinx(\state,\update)$ is satisfied, and that there exists some $\alpha\ge0$ such that
$\propapprox(\bound,\state,\update,\alpha)$ is satisfied.
Then, for all $i\in E$ and $\hat X\in\U$,
\[f(i,\Phi(\hat X_{<i}))\le (m-i+1)\alpha+\max\{\lopt(X)-\lopt(X_{<i}):X\in\U,\,\Phi(X_{<i})=\Phi(\hat X_{<i})\}.\]
\end{theorem}
\begin{proof}
To prove the theorem, we define a new function $f_X$ in which the outcome of the $\max$ operation
in $f$ is decided by a special parameter $X$ (rather than actually choosing the maximum).
\begin{align*}
f_X(i,r,s)=\begin{cases}
0&\text{if }i>m\\
f_X(i+1,r,s)&\text{if }i\le m\text{ and }i\notin X\\
f_X(i+1,\Pi(i,r,s))+\bound(i,r,s)&\text{if }i\le m\text{ and }i\in X\\
\end{cases}
\end{align*}
Let $\hat X\in\U$ and $i\in E$.
By recursively expanding out the definition of $f$, we can see that 
\[f(i,\Phi(\hat X_{<i}))=\sum_{j=1}^k\bound(x_j,r_j,s_j)\] for some $k$ and some sequences
$x_1,\dots,x_k\in E$ and $(r_1,s_1),\dots,(r_k,s_k)\in\{0,\dots,C\}\times T$ with $i\le x_1<\dots<x_k$ and $(r_1,s_1)=\Phi(\hat X_{<i})$.
There may be multiple valid choices of these sequences.
We may assume that for all $j\in\{1,\dots,k\}$, \[f(x_j+1,r_j,s_j)\ne f(x_j+1,\Pi(x_j,r_j,s_j))+\bound(x_j,r_j,s_j)\]
because if not, there is an alternative choice for the sequences which exclude $x_j$ and $(r_j,s_j)$, respectively.
Let $X=\hat X_{<i}\cup \{x_1,\dots,x_k\}$. Then, by definition we have
that $f(x_j,r_j,s_j)=f_X(x_j,r_j,s_j)$ for all $j\in\{1,\dots,k\}$.
Intuitively, $X$ encodes the `decisions' made in $f(i,\Pi(\hat X_{<i}))$.

Note that since $(r_1,s_1)=\Phi(\hat X_{<i})=\Phi(X_{<i})$, by $\propiinx(\Phi,\Pi)$, $(r_2,s_2)=\Pi(x_1,\Phi(X_{<x_1}))=\Phi(X_{<x_2})$
and inductively for all $\ell$,
\begin{equation}
\label{eqn:stateEquiv}
(r_{\ell+1},s_{\ell+1})=\Pi(x_\ell,\Phi(X_{<x_\ell}))=\Phi(X_{<x_{\ell+1}}).
\end{equation}
We claim that $X\in\U$.
If not, then by \cref{eqn:stateEquiv} and the fact that $\hat X\in\U$, there is some $j$
such that $c_{x_j}>C-c(X_{<x_j})=r_j$.
However,
\begin{align*}
x_j\in X &\implies f_X(x_j,r_j,s_j)=f_X(x_j+1,\Pi(x_j,s_j))+\bound(x_j,r_j,s_j)\text{, and}\\
c_{x_j}>r_j&\implies f(x_j,r_j,s_j)=f(x_j+1,r_j,s_j).
\end{align*}
This contradicts $f(x_j,r_j,s_j)=f_X(x_j,r_j,s_j)$ given our earlier observation
that we may assume $f(x_j+1,r_j,s_j)\ne f(x_j+1,\Pi(x_j,r_j,s_j))+\bound(x_j,r_j,s_j)$.
Therefore, $X\in\U$.
We can now finish the proof.
\begin{align*}
f_X(i,\Phi(\hat X_{<i}))&=\sum_{j=1}^k\bound(x_j,r_j,s_j)\tag{by the definition of $f_X$}\\
	&=\sum_{j=1}^k\bound(x_j,\Phi(X_{<x_j}))\tag{by \cref{eqn:stateEquiv}}\\
	&\le k\alpha+\sum_{j=1}^k\lopt(X_{<x_j}\cup\{x_j\})-\lopt(X_{<x_j})\tag{by $\propapprox(\bound,\state,\update,\alpha)$}\\
	&=k\alpha+\sum_{j=1}^k\lopt(X_{<x_j+1})-\lopt(X_{<x_{j}})\tag{by the definition of $X$}\\
	&=k\alpha+\lopt(X_{<x_k+1})-\lopt(X_{<x_1})\tag{since the sum telescopes}\\
	&=k\alpha+\lopt(X)-\lopt(X_{<i})\tag{by the definition of $X$}\\
	&\le (m-i+1)\alpha+\max\{\lopt(X')-\lopt(X'_{<i}):X'\in\U,\,\Phi(\hat X_{<i})=\Phi(X'_{<i})\}
	\tag{since $X\in\U$, $\hat X_{<i}=X_{<i}$, and $k\le m-i+1$.}
\end{align*}
\end{proof}

We remark that while this guarantee may appear to be quite poor in the worst case
(losing an additive factor of $m$), the strength of this approach is that it computes
bounds for many subproblems, using only the amount of time required to compute
$\bound(i,r,s)$ for each subproblem. These subproblem
bounds can be precomputed and efficiently accessed by our branch-and-bound scheme.
Furthermore, for graphic matroids,
our upper bounds are typically very strong in practice, outperforming other methods such as the LP relaxation of the extended formulation \cite{wei2021integer,bazgan2012efficient}
or Lagrangian relaxation \cite{linhares2017improved}. For partition matroids,
our upper bound actually solves the interdiction problem exactly (i.e., we can take $\alpha=0$
in \cref{thm:approx}).

The techniques in this section can be generalized considerably beyond the given setting of matroid interdiction.
We already noted that $\L(X)$ does not necessarily need to be the set of bases of a matroid which exclude elements in $X$;
it suffices that an appropriate function $\delta$ can be computed.
However, $\U$ does not need to be $\{X\se E:c(X)\le C\}$ either; the only important quality that $\U$ should have
is that it should be possible to keep track of whether a set is feasible under additions of elements to the set
using a `small' amount of state (we want it to be `small' so that memory usage is reasonable).
In the above, we used the variable $r$ to do this. As an example generalization, if $\U=\{X\se E:c^1(X)\le C^1,c^2(X)\le C^2\}$,
then we could use variables $r_1\in\{0,\dots,C^1\}$ and $r_2\in\{0,\dots,C^2\}$ to keep track of feasibility.
The proofs require only trivial modifications.

\subsection{Uniform matroids}
\label{sec:uniform}

We start with perhaps the easiest nontrivial case: uniform matroids.
Recall that the independent sets of a uniform matroid over $E$
are defined by $\{S\se E:|S|\le k\}$ for some given $k$.
Intuitively, to determine the replacement
element when an element $e$ is interdicted in a uniform matroid, we only
need to know the number of elements before $e$ which were interdicted,
because independence of a set $S$ depends only on $|S|$.
The rest of this section is spent formalizing this intuition
and developing an exact algorithm for uniform matroid interdiction.

We now define a set $T$ and the functions $\bound$, $\update$, and $\state$, as discussed
in the previous subsection.
Let $T=\{0,1,\dots,m\}$.
We define $\update:E\times T\to T$ and $\state:\U\to T$ as follows:
\begin{align*}
\update(i,n)&=n+1,\\
\state(X)&=|X|.
\end{align*}
Here, the state $T$ keeps track of the number of interdicted elements,
in line with our intuition.
It is trivial to see that $\propiinx(\state,\update)$ is satisfied.
The proof of the following theorem leads to our definition of $\bound$
and the proof of $\propbound(\bound,\state,\update)$.
\begin{theorem}
\label{thm:uniformBound}
Suppose that $M$ is uniform.
Then there exists a function $\bound(i,r,n)$ satisfying
$\propapprox(\bound,\state,\update,0)$ with equality, i.e., for all $X\in\U$, $r\in\{0,\dots,C\}$ and $i\in E$,
\begin{align*}
c_i\le C-c(X_{<i})\implies\bound(i,C-c(X_{<i}),\state(X_{<i})))
=F(X_{<i}\cup\{i\})-F(X_{<i})).
\end{align*}
\end{theorem}
\begin{proof}
Consider some $X\se\{1,\dots,i-1\}$ such that $X\cup\{i\}\in\U$ and $\state(X)=n$.
We want to determine $F(X\cup\{i\})-F(X)$.
Let $B=\bmin(X)$ be the lexicographically smallest basis of $E\sm X$.
As discussed in \cref{sec:intro}, the lexicographically smallest basis is always of minimum weight,
given our item ordering assumption.

If $i\notin B$, then $F(X\cup\{i\})=F(X)$ because
$B$ is also the lexicographically smallest basis of $E\sm(X\cup\{i\})$.
Otherwise, if $i\in B$, then there are two cases.
If there is no $j>i$ with $j\notin B$, then 
$F(X\cup\{i\})-F(X)=\infty$ because there is no
element that can replace element $i$ in $B$.
Otherwise, if such a $j$ exists, then the minimum such element $j$ replaces element $i$,
so $F(X\cup\{i\})-F(X)=w_j-w_i$.

We claim that these conditions can be tested for only by knowing
$i$, $k$ and $n$ (recall $n=|X|$).
Observe that $B$, the lexicographically smallest basis of $E\sm X$, is defined by
$B=\{j\in E:j\notin X,\,j\le k+n\}$ because uniform matroids are connected so, by \labelcref{ass:order},
the elements are sorted by weight.
Let $j=k+n+1$.
There are three cases.
\begin{enumerate}
\item If $j\le i$, then $i\notin B$ because $i\ge j=k+n+1$. So, we take $\bound(i,r,n)=0$.
\item If $i<j\le |E|$, then $i\in B$ because $i\notin X$ and $i<j=k+n+1$.
		We have $j\in E\sm B$ by definition, and for any $i<j'<j$ 
		we know $j'\in B$. So $B\sm\{i\}\cup\{j\}$
		is the lexicographically smallest basis of $E\sm (X\cup\{i\})$. Hence, we take
		$\bound(i,r,n)=w_j-w_i$.
\item If $j>|E|$, then again $i\in B$, but there is no $j'>i$ with $j'\in E\sm B$.
	So, there is no element that can replace $i$, and hence we take $\bound(i,r,n)=\infty$.
	\qedhere
\end{enumerate}
To summarize, $\bound(i,r,n)$ is defined as follows. Note that it does not actually depend on $r$.
\begin{align*}
\bound(i,r,n)=\begin{cases}
0&\text{if }k+n< i,\\
w_{k+n+1}-w_i&\text{if }i\le k+n<|E|,\\
\infty&\text{if }|E|\le k+n.\\
\end{cases}
\end{align*}
\end{proof}
Property $\propbound(\bound,\state,\update)$ follows easily from this. Let $i\in E$ and $s\in T$.
If there exists some $X\in\U$ such that $C-c(X_{<i})=r$ and $\state(X_{<i})=s$, then by \cref{thm:uniformBound}
we have
\begin{align*}
c_i\le r\implies\bound(i,r,s)&=F(X_{<i}\cup\{i\})-F(X_{<i})\\
	&=\max\begin{Bmatrix*}[l]
	F(X'_{<i}\cup\{i\})-F(X'_{<i})\\
	\text{s.t. }X'\in\U,\,C-c(X_{<i})=r,\,\state(X'_{<i})=s
\end{Bmatrix*}.\end{align*}
Otherwise, if no such $X$ exists, then the $\max$ in $\propbound(\bound,\state,\update)$ is taken
over the empty set, so the condition is vacuously satisfied.

Since uniform matroids are so simple, we can further show that with $\bound$
defined as above, $F(\emptyset)+f(1,C,\state(\emptyset))$ is not just an
upper bound, but actually solves the matroid interdiction problem exactly.
\begin{corollary}
	For uniform $M$, we have 
	\begin{align*}
	F(\emptyset)+f(1,\state(\emptyset))=\max\{F(X):X\in\U\}=\OPT.
	\end{align*}
\end{corollary}
The proof follows immediately from \cref{thm:uniformBound,thm:ubCorrect,thm:approx}.
The following is then also immediate from the observation that $\bound$
can be computed in $O(1)$ time and $|T|=O(m)$, so the DP 
table for $f$ has size $O(m^2 C)$.
\begin{corollary}
	Uniform matroid interdiction can be solved in time $O(m^2C)$.
\end{corollary}
Finally, to complement this, we prove that the decision problem version
of uniform matroid interdiction is NP-complete, and hence no better
running time than pseudopolynomial is possible unless $\operatorname{P}=\operatorname{NP}$.
We define this decision problem as follows: given an instance of uniform matroid
interdiction and an integer $t$, determine if there is a feasible solution
to the uniform matroid interdiction problem of value at least $t$.
\begin{theorem}
The uniform matroid interdiction decision problem is weakly NP-complete.
\end{theorem}
\begin{proof}
The problem is easily seen to be in NP, because determining feasibility requires
only checking whether the knapsack constraint is satisfied and determining
the minimum basis weight in a uniform matroid.

Given a knapsack instance with $n$ items, costs $a$, budget $C$ and profits $p$,
we define a uniform matroid $M=(E,\I)$ where $E=\{1,2,\dots,2n\}$, $\I=\{S\se E:|S|\le n\}$, and we define
weights and interdiction costs as follows.
Let $p_{\max}$ be the largest item profit from $p$.
For $e\in\{1,\dots,n\}$, set $w_e=p_{\max}-p_e$ and $c_e=a_e$.
For $e\in\{n+1,\dots,2n\}$ set $w_e=p_{\max}$ and $c_e=C+1$.
The uniform matroid interdiction instance has the same capacity $C$ as the knapsack.
It is easily seen that interdicting any element $e\in\{1,\dots,n\}$ costs $a_e$
and causes the minimum basis weight to increase by $p_e$, because the replacement
element is always after element $n$. No elements after $n$ can be interdicted,
because they have interdiction cost greater than $C$. This is exactly the same change to
profit and cost which are seen when adding $e$ to the knapsack. So, a set $X$ is an optimal
solution for the knapsack problem if and only if $X$ is an optimal solution
to the corresponding uniform matroid interdiction problem.
\end{proof}

In \cref{sec:directsum}, we extend these results to partition matroids, which are the direct
sum of uniform matroids.

\subsection{Direct sums of matroids}
\label{sec:directsum}

In this section, we show how to derive an upper bound for the interdiction problem
on a matroid $M=(E,\I)$ which is the direct sum of two matroids (or more, by induction)
for which we have upper bounds.
Furthermore, any additive approximation bound for the matroids being summed is preserved.
The most immediate consequence of this result is a pseudopolynomial time algorithm
for solving the partition matroid interdiction problem.
This result is fairly easy to intuitively understand, but a rigorous proof ends up being
quite technical due to the desire to keep the number of DP states small.

Let $m=|E|$, and let $M_1=(E,\I_1)$ and $M_2=(E,\I_2)$ be matroids.
We assume without loss of generality that there exists
some $p$ such that if $J\in\I_1$ then $J\se\{1,\dots,p\}$
and if $J\in\I_2$ then $J\se\{p+1,\dots,m\}$, i.e., that the matroids are disjoint.
We assume that $M$ is represented as follows:
\begin{align*}
	E&=\{1,\dots,m\},\\
	\I&=\{J:J\cap\{1,\dots,p\}\in\I_1\text{ and }J\cap\{p+1,\dots,m\}\in\I_2\}.
\end{align*}
We assume that the subsets $\{1,\dots,p\}$ and $\{p+1,\dots,m\}$ of $E$
are individually sorted as necessary for the constructions on the individual matroids.
For example, if $M_1$ and $M_2$ are both uniform or graphic, then we would have
$w_1\le w_2\le\dots\le w_{p}$ and $w_{p+1}\le w_{p+2}\le\dots\le w_{m}$.
Notice that this satisfies our element ordering assumption \labelcref{ass:order} since
there are no two elements $i\in\{1,\dots,p\}$ and $j\in\{p+1,\dots,m\}$
which are in the same component of $M$.
Let $F_M(X)$ be the minimum $w$-weight of a basis $B$ of $E\sm X$ in $M$, or $\infty$
if no such basis exists.

Suppose that $T_i$, $\state_i$, $\update_i$ and $\bound_i$, correspond to
$M_i$ for $i\in\{1,2\}$ and that they satisfy $\propbound(\bound_i,\state_i,\update_i)$
and $\propiinx(\state_i,\update_i)$.
In this section we make the additional assumption that for all $S\se\{1,\dots,p\}$,
\begin{align}
&\state_2(S)=\state_2(\emptyset).\label{eqn:dsassume}
\end{align}
This is a natural assumption to make because elements
$1$ through $p$ are not contained in any independent set of $M_2$, and hence should be
irrelevant to $\state_2$. This assumption arises because we represent the matroids
as having the same ground set, which simplifies the main proof.
Let $T=T_1\sqcup T_2$ (the disjoint union of $T_1$ and $T_2$).
We define the functions $\state$, $\update$, and $\bound$ as follows.
\begin{align*}
\state(X)&=
	\begin{cases}
		\state_1(X)&\text{if }X\se\{1,\dots,p\}\\
		\state_2(X)&\text{otherwise}
	\end{cases}\\
\update(i,s)&=
	\begin{cases}
		\update_1(i,s)&\text{if }i\le p\\
		\update_2(i,\state_2(\emptyset))&\text{if }i>p\text{ and }s\in T_1\\
		\update_2(i,s)&\text{if }i>p\text{ and }s\in T_2
	\end{cases}\\
\bound(i,r,s)&=
	\begin{cases}
		\bound_1(i,r,s)&\text{if }i\le p\\
		\bound_2(i,r,\state_2(\emptyset))&\text{if }i>p\text{ and }s\in T_1\\
		\bound_2(i,r,s)&\text{if }i>p\text{ and }s\in T_2
	\end{cases}
\end{align*}
We remark that while these functions may appear to be asymmetrical, the desired properties hold
independent of which matroid is chosen as $M_1$ or $M_2$.
The asymmetry is necessary to handle the fact that in the recursive definition of $f$,
a state $s$ can only transition upon an item being interdicted. So, when we transition from
processing elements of $M_1$ to elements of $M_2$, we may still have a state corresponding
to $M_1$ if the boundary element is not interdicted, and hence we need to transition to a
state from $M_2$ at the earliest point that an element from $M_2$ is interdicted.
This construction can be simplified if we allow state transitions in $f$ when items are not interdicted, but it comes at the cost of heavier notation in the rest of the paper.
The following establishes that this works as intended.
\begin{theorem}
\label{lem:directsum}
If $T$, $\state$, $\update$, and $\bound$ are defined as above, then
$\propbound(\bound,\state,\update)$ and $\propiinx(\state,\update)$ are satisfied.
Furthermore, if there exists some $\alpha\ge0$ such that
$\propapprox(\bound_1,\state_1,\update_1,\alpha)$ and $\propapprox(\bound_2,\state_2,\update_2,\alpha)$ are satisfied, then
$\propapprox(\bound,\state,\update,\alpha)$ is satisfied.
\end{theorem}
The proof is technical and not very interesting, so we defer it to \cref{apx:directsum}.
Since partition matroids are the direct sum of uniform matroids,
we can then establish the following.
\begin{corollary}
	\label{cor:solvepartition}
	The partition matroid interdiction problem can be solved in time $O(km^2C)=O(m^3C)$,
	where $k$ is the number of disjoint uniform matroids that need to be summed to
	construct the partition matroid.
\end{corollary}
\begin{proof}
Applying our method to construct a state set $T$ for a partition matroid
yields $T=\{0,\dots,C\}\times(T_1\sqcup T_2\sqcup\dots\sqcup T_k)$ where
the $T_i$ correspond to uniform matroids.
Any element in the set $T_1\sqcup\dots\sqcup T_k$ can be addressed
by using a number in the set $\{1,\dots,k\}$ to indicate which $T_j$ the element is
in, along with a number from $\{1,\dots,m\}$ to represent the element itself in $T_j$.
So, considering the $i$ parameter of $f$, which
is in $\{1,\dots,m\}$, and the $r$ parameter, which is in $\{0,\dots,C\}$,
the total number of DP states of $f$ is $O(km^2C)$,
and clearly each can be computed in constant time.
\end{proof}

\subsection{Graphic matroids}

In this section, we describe our upper bound for the case when $M$ is graphic.
This case is also known as the MST interdiction problem
(e.g., see \cite{wei2021integer,linhares2017improved}),
or the most vital edges in the MST problem (e.g., see \cite{lin1993most}).

We follow again the framework laid out in \cref{sec:ub}, wherein we must define
a set $T$ and functions $\bound$, $\update$, and $\state$.
For now we define $T=\{\epsilon\}$ (i.e., an arbitrary 1-element set)
as no additional state is needed to get a good
bound; later we will show how to extend $T$ to get stronger bounds.
So, we can simply define $\update(i,\epsilon)=\epsilon$ and $\state(X)=\epsilon$. It is trivial
to verify the desired conditions hold.
To simplify notation, we simply write $\bound(i,r)$ to mean $\delta(i,r,\epsilon)$.

It remains to define a function $\bound$ satisfying $\propbound(\bound,\state,\update)$.
However, we will not actually use the condition that $\state(X_{<i})\ge r$ from the definition of $\propbound$,
so our function $\bound$ will actually satisfy the following stronger condition:
\begin{align}
\bound(i,r)\ge \max\left\{F(X_{<i}\cup\{i\})-F(X_{<i}):X\in\U,\,\state(X_{<i})\le r\right\}.
\label{eqn:gFnCond2}
\end{align}
\cref{alg:computeGFnImproved} computes the value of $\bound(i,r)$ for all values
of $i$ and $r$.
In the algorithm, \MinCut{$E,c,u,v$} computes the minimum $u$-$v$ cut
in the subgraph of $G$ with edges $E$ and edge weights $c$.

\begin{algorithm}[t]
	Let $\bound(i,r)$ be defined as a table\;
	\For{$i=1,\dots,m$}{
		\lFor{$r=0,\dots,C$}{$\bound(i,r)\gets \infty$}
		\If{there exists a replacement element $k$ for $i$ in $M\sm\{1,\dots,i-1\}$}{
			\lFor{$r=0,\dots,C$}{$\bound(i,r)\gets w_k-w_i$}\label{algline:init}
		}
		$E'\gets\emptyset, c'\gets0$\;
		\For{$j=1,\dots,i-1$}{\label{algline:initloop}
			$E'\gets E'\cup\{j\}$\;
			$c'_j\gets c_j$\;
		}
		Let $u$ and $v$ be the endpoints of edge $i$\;
		$x\gets\MinCut{E',c',u,v}$\;
		\lFor{$r=\max\{0,C-x+1\},\dots,C$}{$\bound(i,r)\gets 0$}\label{algline:zerobound}
		\For{$j=i+1,\dots,m$}{\label{algline:innerloop}
			$E'\gets E'\cup\{j\}$\;
			$c'_j\gets \infty$\;
			$x\gets\MinCut{E',c',u,v}$\;
			\lFor{$r=\max\{0,C-x+1\},\dots,C$}{$\bound(i,r)\gets\min\{\bound(i,r),w_j-w_i\}$}\label{algline:setbound}
			\lIf{$x>C$}{\Break}
		}
	}
    \caption{Algorithm to compute $\bound(i,r)$ for all $i$ and $r$, when $M$ is graphic.}
	\label{alg:computeGFnImproved}
\end{algorithm}

\begin{theorem}
The values of $\bound(i,r)$ after running \cref{alg:computeGFnImproved} satisfy \cref{eqn:gFnCond2}.
\end{theorem}
\begin{proof}
Notice that each iteration $i$ of the outer loop can be computed independently of the
other iterations. So, fix $i$ arbitrarily.
If there exists a replacement element $k$ for $i$ in $M\sm\{1,\dots,i-1\}$, then on \cref{algline:init}, the algorithm sets $\bound(i,r)$ to $w_k-w_i$ for every $i$ and $r$,
This satisfies the desired inequality because it effectively assumes $X=E$, which does not necessarily
satisfy $X\in\U$ and $\phi(X_{<i})\le r$, but certainly upper bounds $F(X_{<i}\cup\{i\})-F(X_{<i})$.
In fact, this simple bound for $\bound(i,r)$ works with any matroid, not just graphic matroids.

Consider the case where the algorithm
sets $\bound(i,r)=0$ on \cref{algline:zerobound}.
Notice that we set $\bound(i,r)=0$ precisely when $C-r<\MinCut{E',c',u,v}$.
In order for $i$ to be an edge in the lexicographically smallest spanning tree,
we must cut $u$ and $v$ in $(V,E')$; otherwise,
an earlier edge could replace $i$ in the lexicographically smallest spanning tree. So, if $C-r<\MinCut{E',c',u,v}$,
then there is no $X$ for which $i$ is an edge in the lexicographically smallest spanning tree of $E\sm X$, and hence 
there will be no gain from interdicting $i$. That is,
$F(X\cup\{i\})-F(X)=0$ for all $X\se\{1,\dots,i-1\}$ with
$c(X)\le C-r$. Therefore, the condition is satisfied.

We now show that the assignment to $\bound(i,r)$ on \cref{algline:setbound} is correct.
Let $X\se\{1,\dots,1-i\}$ be such that $c(X)\le C-r$.
Consider the lexicographically smallest spanning tree $T$ of $G\sm X$. If edge $i$ is not in $T$, then
$F(X\cup\{i\})-F(X)=0$, and since $w_j\ge w_i$,
we set $\bound(i,r)\ge0$ and the condition is satisfied.

So, we may assume $i$ is in $T$.
Fix $r$, let $j$ be the earliest iteration 
where we set $\bound(i,r)=w_j-w_i$ on \cref{algline:setbound},
and let $E'$ and $c'$ be as they are when \cref{algline:setbound} executes on iteration $j$.
Since we set $\bound(i,r)=w_j-w_i$, there is no $u$-$v$ cut of cost at most $C-r$
in the graph $(V,E')$ with edge costs $c'$.
However, in order for us to have $i\in T$, $X$ must be a $u$-$v$ cut in $(V,E'\cap\{1,\dots,{i-1}\})$; otherwise
$i$ could be replaced in $T$ with another edge $e'$ from the cut with $e'<i$, yielding an lexicographically smaller spanning tree.
However, since $c(X)\le C-r$ and every $u$-$v$ cut in $(V,E')$ has cost greater than $C-r$,
in particular $X$ is not a $u$-$v$ cut in $(V,E')$.
So, some edge $e\in\{{i+1},\dots,j\}$ crosses the cut defined by $X$.
Therefore, adding $e$ to $T\sm\{i\}$ produces a spanning tree,
and $F(X\cup\{i\})-F(X)\le w_j-w_i$ as desired.
\end{proof}

\subsubsection{Worst-case performance}
\label{sec:mstworstcase}
We have now established how to compute some upper bounds for graphic matroid interdiction,
but the above algorithm and proof gives little intuition about how strong the bounds are.
In this section we define a family of instances for which $f(1,\state(\emptyset))=\Theta(m)\OPT$.
Although this result is negative---for example, the Lagrangian upper bound is known to be at most
$2\cdot \OPT$ \cite{linhares2017improved}---we argue in \cref{sec:comp} that
the poor worst-case performance is outweighed by other attractive features of this
bound in the context of a branch-and-bound scheme.

The intuition is as follows.
For each $k\ge1$ and $M\ge1$, we define a {\em comb} instance $\mathcal{I}(k,M)$.
Each comb instance can be divided into $k$ sections, called {\em teeth}, which act like independent subproblems,
in the following sense: the choice of which edges to include in the spanning tree
from any one of the $k$ sections has no impact on which edges can be added in any of the other sections.
In each tooth, there is a high-cost edge and a low-cost edge which must both
be interdicted to significantly increase the MST weight.
However, in the upper bound computation, after the leader pays for a single high-cost edge from one
of the teeth, it can then significantly increase the MST weight in every tooth by interdicting only the
low-cost edge, because there is no way to distinguish which tooth the interdicted high-cost edge is from.
The instances $\mathcal{I}(k,M)$ are formally described in \cref{fig:badGuy}.

\begin{figure}
	\centering
	\includegraphics[height=2in]{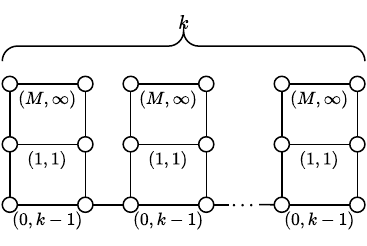}
	\caption{Comb instance with $k$ `teeth'.
	 The unlabelled edges have $w_e=0$ and $c_e=\infty$;
	 all other edges are labelled with $(w_e,c_e)$.
	 The capacity $C$ is $2k-1$.}
	\label{fig:badGuy}
\end{figure}

\begin{proposition}
The optimal solution to the MST interdiction problem on a comb instance $\mathcal{I}(k,M)$
has weight $M+1$.
\end{proposition}
\begin{proof}
Notice that of the edges with non-infinite interdiction cost, they all either have
cost $1$ or cost $k-1$, and there is one of each in each tooth.
Furthermore, in each tooth, it is necessary to interdict the edge of cost $k-1$
to gain any advantage from interdicting the edge of cost $1$, because otherwise
the cost $1$ edge will not be in the MST.
Since the capacity is $2k-1$, we can interdict at most two of the edges
of cost $k-1$. If we interdict one of them, say in the first tooth,
then we can additionally interdict the edge of cost $1$ in that tooth
which will increase the MST weight from $0$ to $M$.
This leaves $k-1$ capacity remaining, with which we can interdict the $k-1$
cost edge in one more tooth, which increases the MST weight to $M+1$.
\end{proof}

\begin{proposition}
\label{thm:dpWorstCase}
The value of the upper bound $f(1,C)+F(\emptyset)$ on a comb instance $\mathcal{I}(k,M)$
is at least $k(M-1)+1=\Theta(m)\cdot \OPT$.
\end{proposition}
\begin{proof}
Recall that edges are sorted by increasing weight; hence,
all of the weight $0$ edges come first,
followed by the weight $1$ edges, and then the weight $M$ edges.
Notice that the MST of a comb instance has weight $0$ because the subgraph with weight $0$ edges is a spanning tree.

Suppose that in the DP we first interdict some edge $i$ with cost $k-1$.
This causes the MST weight to increase by $1$ (since the edge of weight 1 in that tooth must now be used),
and hence $\bound(i,C)\ge1$.
We claim that interdicting this one edge allows us to interdict
every edge of cost $1$, while gaining an upper bound weight of at least $M-1$ from each edge,
despite the fact weight from interdicting these edges does not contribute any actual weight.
This plan costs
$k-1+k=2k-1=C$, and produces an upper bound of $f(1,C)+F(\emptyset)=1+k(M-1)+0$
as desired.

To show that the claim holds, it suffices to show that for every edge $j$
of cost 1 and weight 1 and for every $1\le r\le C-(k-1)$,
we have $\bound(j,r)\ge M-1$. Let $j'$ be the edge of cost $k-1$ on the same tooth as $j$.
We define $X=\{j'\}$. Then $\sum_{k\in X}c_k=k-1\le C-r$.
Observe that $F(X\cup\{j\})=M$ and $F(X)=1$. This completes
the proof.
\end{proof}
It is easy to see that the maximum spanning tree of a comb instance
has weight $kM$ (this corresponds to including the edge of weight $M$ in
the spanning tree from each of the $k$ teeth). So, \cref{thm:dpWorstCase} shows that
the upper bound is nearly returning the maximum spanning tree weight and
hence the analysis is reasonably tight.

It is interesting to note that the same proof would work even if $\delta$ satisfied $\propbound(\delta,\state,\update)$
with equality. So, the reason why the bound is poor is because
our state space $T$ is not large enough. The existence of a pseudopolynomial
sized $T$ for which the upper bound has worst-case performance of $o(m)\OPT$
is an open question.
While the worst-case performance may appear to be quite poor, our computational
experiments in \cref{sec:comp} suggest that the average case is significantly better.

\subsection{Strengthening}
\label{sec:strengthen}
The strength of our upper bounds depend upon the set $T$ and the ability to
exploit $T$ to get strong bounds on the discrete derivative problem.
In this section we present a generic way to strengthen the bounds,
which comes at an exponential cost in running time and memory usage,
but has the practical advantage that it allows us to compute
stronger bounds for hard instances, which significantly improves
solver performance.

The idea explored in this section is to extend $T$ to store {\em complete}
information about whether the first $p$ elements of the matroid were interdicted,
for some small $p$. Evidently, if $p=m$, then this extended $T$ would completely describe
the interdiction set $X$, and hence the interdiction problem could be solved
exactly just by the upper bound. But, this requires $|T|=\Omega(2^p)$, so
we have to choose $p$ to be small in order for the DP table to be of
a practical size.

Given the set $T$ and functions $\state$, $\update$, and $\delta$ defined
for a matroid $M$, we define $T_p$, $\state_p$, and $\update_p$ for all
$p\in\{0,1,\dots,m\}$ as follows:
\begin{align*}
	T_p&=T\times \{X\cap\{1,2,\dots,p\}:X\in\U\}\\
	\update_p(i,(s,X))&=\begin{cases}
		(\update(i,s), X\cup\{i\})&\text{if $i\le p$}\\
		(\update(i,s), X)&\text{otherwise}\\
	\end{cases}\\
	\state_p(X)&=(\state(X), X\cap\{1,2,\dots,p\})
\end{align*}
The function $\bound_p(i,r,(s,X))$ is a bit less straightforward to define.
In essence, we want to delete all elements of $X$ from the matroid, 
and for all $j\notin X$ with $j\le p$, forbid $j$ from being interdicted.
In \cref{alg:computeGFnImproved}, this amounts to modifying the loop on \cref{algline:initloop}
to skip over elements $j\le p$ such that $j\in X$ while subtracting their capacity from $C$,
and setting elements $j\le p$ such that $j\notin X$ to have infinite cost $c'_j$.
Other matroids must be handled on a case-by-case basis depending on
the definition of $\bound$.

We then define $f_p(i,r,(s,X))$ to be the same as $f(i,r,s)$, but using
the functions $\state_p$, $\update_p$, and $\bound_p$.
It is easy to see that for all $i$, $r$, and $s$, we have
$f_0(i,r,(s,\emptyset))=f(i,r,s)$,
and 
\begin{align*}
f_p(i,C-c(X_{<i}),(\phi(X_{<i}),X_{\le p}))\le f(i,C-c(X_{<i}),\phi(X_{<i}))
\end{align*}
for $p\ge0$ and $X\in\U$. However, $f_p$ is still a valid upper bound
by \cref{thm:ubCorrect}, as desired.

\section{Exact algorithms}
\label{sec:exact}

In this section, we describe a branch-and-bound algorithm for exactly solving the
matroid interdiction problem. The performance of our algorithm crucially depends on
the upper bounds derived in \cref{sec:ub}.
We also use a heuristic greedy lower bound to improve computational performance further.

\subsection{Greedy lower bound}
\label{sec:greedy}

We use a simple greedy heuristic to find an initial feasible solution (i.e., lower bound)
to seed the branch-and-bound algorithm with. Although this heuristic does not
have any performance guarantee, we find that for MST interdiction,
it improves performance in practice
as it is able to find good solutions faster than is possible by branching alone.
In fact, the solution returned by the heuristic has an average optimality gap of 1.48\%
over the entire data set. A full performance analysis can be found in \cref{sec:perfDetails}.

The basic idea behind the heuristic is to repeatedly select an element to interdict which
maximizes the ratio of basis weight increase to interdiction cost. To formalize this,
suppose we start with some feasible solution $(X,Y)$. For $e\in Y$ let $R(Y,e)$ be the replacement element for
$e$ if $e$ is interdicted from $Y$. We could select the element $e\in Y$ to interdict (i.e., add to $X$)
which maximizes $(w_{R(e)}-w_e)/c_e$ (we call this quantity the {\em efficiency}).
$X$ and $Y$ are then updated accordingly and
this process repeated until there is no item which can be added
to $X$ without going over capacity.

This idea is not completely new; a similar algorithm appeared previously
for the case where all costs are 1 \cite{bazgan2012efficient}.
However, we modify the element selection slightly based on analyzing some cases where we noticed
poor performance. The fundamental issue with this algorithm
is that it may be necessary to interdict some element(s) with low efficiency
in order to have the opportunity to interdict an element with large efficiency.
So, by always choosing the item with the highest immediate benefit, we are in some sense
behaving too greedily and could potentially miss out on the largest efficiency gains.

To address this issue, we modify the selection rule to look ahead along {\em replacement chains}.
We define the replacement chain for an element $e\in Y$ as the sequence $r_0,r_1,\dots r_s$ where $r_0=e$
and $r_i=R(Y\cup\{r_0,\dots,r_{i-1}\},r_{i-1})$ for $i\ge1$. The sequence ends when there is
not enough capacity to interdict the next element that would have been be added,
i.e., $s$ is the smallest index such that $c(X\cup\{r_0,\dots,r_s,r_{s+1}\}) > C$.
The replacement chain for an element $e\in Y$ can be computed using the independence oracle in polynomial time
with a single scan through the elements, because we know that $r_0<r_1<\dots<r_s$ due
to our element ordering assumption.

We utilize the replacement chains as follows.
In each step of the heuristic, instead of selecting the element $e\in Y$ which maximizes
$(w_{R(e)}-w_e)/\max\{1,c_e\}$, we instead select the element $e\in Y$ which maximizes
\begin{align*}
	\max\left\{(w_{r_i}-w_{e})/\textstyle\sum_{j=0}^ic_{r_j}:i=1,2,\dots s\right\}.
\end{align*}

\begin{algorithm}[t]
	\Fn{\LowerBound{}}{
		$X\gets\emptyset$\;
		Let $Y$ be a minimum weight basis of $M\sm X$\;
		\While{true}{
			$b\gets -\infty$; $f\gets \bot$\;
			\For{$e\in Y$ such that $c_e\le C-c(X)$}{
				Let $r_0,r_1,\dots,r_s$ be the replacement chain for $e$ in $Y$\;
				$v\gets \max\left\{(w_{r_i}-w_{e})/\textstyle\sum_{j=0}^ic_{r_j}:i=1,2,\dots s\right\}$\;
				\lIf{$b < v$}{$b\gets v$; $f\gets e$}
			}
			\lIf{$f=\bot$}{\Break}
			$X\gets X\cup\{f\}$\;
			$Y\gets Y\sm\{f\}\cup\{R(Y,f)\}$\;
		}
		\Return $(X,Y)$
	}
    \caption{Greedy lower bound.}
	\label{alg:greedy}
\end{algorithm}

The full algorithm is described in \cref{alg:greedy}.
It is easy to see that the output is
a feasible solution because edges are only added
to $X$ when they fit within the capacity, and we maintain that $Y$ is a minimum weight basis
of $M\sm X$ at each iteration.
To integrate this heuristic into our branch and bound scheme, we compute \LowerBound{}
at the beginning of the search, and use it as the initial lower bound.

\begin{figure}
	\centering
	\includegraphics[height=2.5in]{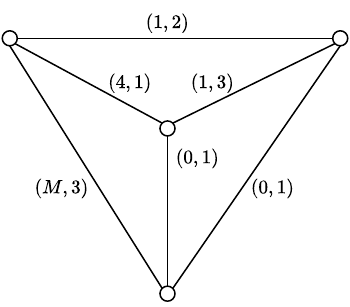}
	\caption{Worst-case instance for the greedy lower bound. Each edge $e$ is labelled with
	$(w_e,c_e)$, and $C=2$.}
	\label{fig:greedyWorstCase}
\end{figure}

In \cref{fig:greedyWorstCase} we show a family of instances parameterized by an integer $M$
for which the heuristic returns an arbitrarily bad solution (with approximation ratio $\Theta(M)$).
Specifically, the greedy lower bound only interdicts the edge labelled $(1,2)$---yielding an MST of weight 4---but
the optimal solution is to interdict both edges labelled $(0,1)$,
which increases the MST weight to $M+2$.
There are countless possible variations of this heuristic and
surely better variants to be found.
However, we found this algorithm to be remarkably effective in
practice---it often finds the optimal solution---so
we did not attempt to improve it further in this work. Discussion regarding the strength of the 
heuristic in practice can be found in \cref{sec:perfDetails}.

\subsection{Branch-and-bound algorithm}

Our branch-and-bound algorithm for solving matroid interdiction exactly is now straightforward to describe, given
the upper and lower bounds from previous sections.
Throughout this section we assume elements are ordered according to \labelcref{ass:order}.
In essence, we recursively enumerate feasible solutions by branching on whether $e\in X$ (i.e., whether $e$ is interdicted)
for $e=1,\dots,m$. At each step, an optimal $Y\in\L(X)$ is maintained by the use of
a dynamic data structure. Since the elements are sorted by non-decreasing weight within each component,
when an element is interdicted from $Y$ it is always replaced by a later element.
So, we do not need to branch on every element: branching on the next element which is also in $Y$ suffices.
The algorithm is formalized in \cref{alg:bnb} and correctness is proved in \cref{thm:bnbworks}.
With branching alone our algorithm is asymptotically faster than the worst-case running time of previous enumerative algorithms
up to polylog factors; this is formalized in \cref{thm:bnbtime}.
However, the biggest improvement in our algorithm over previous algorithms
is the pruning we can accomplish using the upper bounds from \cref{sec:ub}.

\begin{algorithm}[t]
	In a background thread, compute the upper bound $f_p(1,C,\phi(\emptyset),\emptyset)$ with a progressively increasing
	number of prefix bits $p=0,1,\dots$ until the memory usage limit is reached\;
	$X^*,Y^*\gets\LowerBound{}$\;
	$d\gets\DynInterdict.\Init{}$\;
	\Fn{\BranchAndBound{$i$}}{
		\lIf{$w(d.Y) > w(Y^*)$}{$X^*,Y^*\gets (d.X,d.Y)$}
		\lIf{$i=\rank(E)$ or $d.e>m$ or $\min\{c_e:e\ge d.e\}>C-c(d.X)$}{\Return}\label{algline:bnbbasecase}
		\If{some upper bound has finished computation}{
			Let $p=\max\{p:$ computation of $f_p$ has finished$\}$\;
			\lIf{$f_p(d.e, C-c(d.X), d.X\cap\{1,\dots,p\}) + w(d.Y) \le w(Y^*)$}{\Return}\label{algline:boundtest}
		}
		\If{$c_{d.e} \le C-c(d.X)$}{
			\lIf{$d.\Interdict{}=false$}{\Output error}
			$\BranchAndBound{$i$}$\label{algline:takeelem}\;
			$d.\UnInterdict{}$\;
		}
		$d.\Skip{}$\;
		$\BranchAndBound{$i+1$}$\label{algline:ignoreelem}\;
		$d.\UnSkip{}$\;
	}
	\BranchAndBound{$0$}\;
	\Output $(X^*,Y^*)$\;
    \caption{Branch and bound algorithm to exactly solve matroid interdiction.}
	\label{alg:bnb}
\end{algorithm}

The \BranchAndBound function recursively enumerates solutions.
It has one parameter $i$, and refers to global variables
$d$, $X^*$ and $Y^*$.
The parameter $i$ is the number of basis elements that have been skipped (i.e., not interdicted).
When $i=\rank(E)$, there are no more basis elements to interdict.
The global variable $d$ is an instance of \DynInterdict, a dynamic data structure which keeps track
of the optimal lower-level solution as elements are added and removed from the interdiction set,
and the variables $X^*$ and $Y^*$ hold the current best known incumbent solution.
Given an instance $d$ of the \DynInterdict structure, the following variables and operations are available:
\begin{itemize}
\item $d.X$ is the current set of interdicted elements. Initially this is $\emptyset$.
\item $d.Y$ is the lexicographically smallest basis of $M\sm d.X$.
\item $d.e$ returns the current element being considered for interdiction. Initially, this
is the first (smallest weight) element of $d.Y$.
\item $d.\Skip()$ advances $d.e$ to point to the next element of $d.Y$ following $d.e$.
\item $d.\UnSkip()$ undoes a $\Skip()$ operation.
It is only valid to call this if the last (non-undo) operation was a $\Skip()$ operation.
\item $d.\Interdict()$ adds $d.e$ to $d.X$, and replaces $d.e$ in $d.Y$ with the replacement
element for $d.e$ (given the current $d.X$).
If no such replacement element exists, the function should return false.
It then advances $d.e$ to point to the smallest
element of $d.Y$ after the original $d.e$ (note that this next element is not necessarily the replacement
element for the original $d.e$).
\item $d.\UnInterdict()$ undoes an $\Interdict()$ operation.
It is only valid to call this if the last (non-undo) operation was an $\Interdict()$ operation.
\end{itemize}

We will describe the actual implementation of this data structure that we use later; for now
it is enough to know that the operations satisfy this description.
We now prove the two main results of this section. \cref{thm:bnbworks} establishes
correctness of the algorithm, and \cref{thm:bnbtime} establishes a worst-case
running time bound.

\begin{theorem}
\label{thm:bnbworks}
\cref{alg:bnb} correctly solves the matroid interdiction problem.
\end{theorem}
\begin{proof}
Before \BranchAndBound{$0$} is called, $X^*$ and $Y^*$ are initialized to a feasible
solution by the \LowerBound{} function. We show that the algorithm maintains the
invariant that $(X^*,Y^*)$ is feasible, and for every feasible solution,
the algorithm either enumerates it or proves that it is no better than the current
best known solution $(X^*,Y^*)$. It follows that the algorithm always outputs an optimal solution.

Assuming the correct implementations of the operations on the \DynInterdict structure,
it is easy to see that an element is only interdicted (added to $d.X$) if there is enough capacity
to do so, and it is maintained that $d.Y$ is the lexicographically smallest basis of $M\sm d.X$.
So, $(d.X,d.Y)$ is always a feasible solution.
Since $(X^*,Y^*)$ is only ever updated to hold the value of $(d.X,d.Y)$,
$(X^*,Y^*)$ also remains feasible at every step of the algorithm.

Now, fix some $X\in\U$, and observe that if $Y=\bmin(X)$,
then for any $X'$ such that $X\se X'\se E\sm Y$, we also have $Y=\bmin(X')$.
So, since $c\ge0$, it suffices to show that the algorithm enumerates
all $X'$ such that $\bmin(X')\ne\bmin(X)$ for every $X\subset X'$. We call such $X'$
{\em undominated}.

For now, assume that no upper bound has been computed, so that
the bound test on \cref{algline:boundtest} does not prune any nodes. We also assume that
the leader is unable to reduce the rank of the matroid for any $X\in\U$,
that is, $\bmin(X)$ always exists. We first
show that with these assumptions, the algorithm enumerates all feasible undominated
sets. After establishing this, we prove that the bound test only prunes a node
when enumerating it could not lead to an improvement in the incumbent solution $(X^*,Y^*)$,
and that if the leader is able to reduce the matroid rank, this situation is detected.

Suppose $X'=\{x_1,x_2,\dots,x_q\}$ is undominated, where $x_1<x_2<\dots<x_q$.
We claim that $X'_p:=\{x_1,x_2,\dots,x_p\}$ is also undominated for all $p\in\{0,1,\dots,q\}$.
Assume for the sake of contradiction that $X'$ is undominated but $X'_{q-1}=X'\sm\{x_q\}$ is not.
So, there is some $f\in X'\sm\{x_q\}$ such that
$\bmin(X'\sm\{x_q,f\})=\bmin(X'\sm\{x_q\})$. This implies that $f$ is not in the lexicographically smallest
basis $\B$ of $M\sm(X'\sm\{x_q\})$ (otherwise, interdicting it would cause the basis to change).
But, if $x_q$ is interdicted from $\B$, it will be replaced with an element after $x_q$ because
the elements are sorted by weight within each component. In particular, $x_q$ cannot be replaced
with $f$. So, $f$ is also not in the lexicographically smallest basis of $M\sm X'$, but this implies
that $\bmin(M\sm X')=\bmin(M\sm(X'\sm f))$, contradicting that $X'$ is undominated.

Assume that $X'\in\U$. We claim that the algorithm enumerates $X'$;
that is, there is some point where \BranchAndBound is called when $d.X=X'$.
We show by induction on $p$ that each $X'_p$ is enumerated.
Clearly, $X'_0=\emptyset$ is enumerated, because $\emptyset$ is the initial
value of $d.X$. Now suppose $p\ge1$ and that $X'_{p-1}$ is enumerated.
Consider the values of $d$ and $i$ when $X'_{p-1}$ is enumerated; if it is enumerated multiple times
then break the tie by choosing $i$ as small as possible.
So, $d.X=X'_{p-1}$, $d.Y$ is the lexicographically smallest basis of $M\sm X'_{p-1}$, and, because we chose $i$ as small
as possible, $d.e$ is the
smallest element of $d.Y$ which comes after $x_{p-1}$.
Since $X'_{p}$ is undominated, $\bmin(X'_p)\ne\bmin(X'_{p-1})$.
So, $x_p$ must be in $d.Y$; otherwise, $d.Y$ would also be the lexicographically smallest basis of $M\sm X'_p$,
contradicting that it is undominated. Hence, $d.e\le x_p$. If $d.e=x_p$, then we are done, because
the algorithm will enumerate $X'_p$ via the recursive call on \cref{algline:takeelem}.
Therefore, if we don't already have $d.e=x_p$
then there is a sequence of recursive calls on \cref{algline:ignoreelem} that advance
$d.e$ to point to $x_p$, at which point the recursive call on \cref{algline:takeelem}
will enumerate $X'_p$.

So, without the bound test on \cref{algline:boundtest}, the algorithm would have enumerated
every undominated $X'\in\U$. It remains to show that if some $X'$ is pruned
by the bound test, it could not have led to an improved solution $(X^*,Y^*)$.
Suppose that for some $X^*$, $Y^*$, $d.X$, $d.Y$, and $d.e$ we have $f_p(d.e,C-c(d.X),d.X\cap\{1,\dots,p\})+w(d.Y)\le w(Y^*)$.
Recall that $f_p(d.e,C-c(d.X),d.X\cap\{1,\dots,p\})$ is an upper bound on how much the minimum basis weight $w(d.Y)$
can increase by adding edges from $\{d.e,d.e+1,\dots,m\}$ to $d.X$.
So, if the bound test prunes this node, then no child of this node can have minimum basis
weight larger than $Y^*$, as desired.

Finally, we claim that if there is some $X\in\U$ such that there is no basis using elements
in $E\sm X$, then the algorithm correctly detects this case and outputs an error.
First, notice that if there is such an $X$, we may assume it is minimal with this property, and hence undominated.
Then, deleting the largest element from $X$ yields an undominated set $X'$, so by the same argument
used earlier, $X'$ is enumerated by the algorithm. When \Interdict{} is invoked during
the recursive call where $X'$ is enumerated, it will return false because there is no replacement
edge for $d.e$ (otherwise, $\bmin(X)$ would exist). So, the algorithm will raise an error.
\end{proof}

\begin{theorem}
\label{thm:bnbtime}
Assume that the elements of $M$ are already sorted according to \labelcref{ass:order}.
Then the worst-case running time of \BranchAndBound{$0$} in \cref{alg:bnb} is
\[
	O\left(\binom{\rank(E)+k}{\min\{\rank(E),k\}}\cdot D\right)+I
\]
where $k=\max\{|X|:X\in\U\}$, $D$ is the time required for each operation of the dynamic data structure,
and $I$ is the time required to initialize the dynamic data structure.
\end{theorem}
\begin{proof}
We claim that the total number of leaves in the recursion tree is at most
$t(\rank(E),k)$ where $t(0,j)=t(i,0)=1$,
and $t(i,j)=t(i-1,j)+t(i,j-1)$ for $i,j\ge1$.
To see this, consider the following.
\begin{itemize}
\item The $i$ parameter of \BranchAndBound is incremented
whenever we recurse on \cref{algline:ignoreelem}, and when $i=\rank(E)$, the
recursion terminates. 
\item Since $k=\max\{|X|:X\in\U\}$, the algorithm can recurse at most $k$ times
on \cref{algline:takeelem} before we are guaranteed that $\min\{c_e:e\ge d.e\}>C-c(d.X)$,
and the recursion terminates in this case. We assume that $\min\{c_e:e\ge f\}$
is precomputed for all $f\in E$, so that this test takes $O(1)$ time.
\end{itemize}
So, we can think of the recursive call on \cref{algline:ignoreelem} as decrementing
the $i$ parameter of $t(i,j)$, and the recursive call on \cref{algline:takeelem}
as decrementing the $j$ parameter. Hence, the total number of leaves is at most
$t(\rank(E),k)$.

Now, we show that $t(i,j)=\binom{i+j}{\min\{i,j\}}$.
If $i=0$ or $j=0$, this is clearly the case.
So, suppose $i,j\ge1$. Then, by induction,
\begin{align*}
	t(i,j)&=t(i-1,j)+t(i,j-1)
		=\binom{i+j-1}{\min\{i-1,j\}}+\binom{i+j-1}{\min\{i,j-1\}}=\binom{i+j}{\min\{i,j\}},
\end{align*}
where the last equality follows by splitting into cases depending on whether $i$ or $j$
is larger and applying the recurrence relation for binomial coefficients.
Since the number of nodes in a binary tree is at most 2 times the number of leaves,
the total number of recursive calls is also $O(t(i,j))$.
Each call takes $O(D)$ time (ignoring time taken in child nodes),
so the total running time is as desired.
\end{proof}

In the case of MST interdiction,
by using the best-known decremental dynamic MST data structure (i.e., only supporting
edge deletion operations) of Holm et al \cite{holm2001poly},
we can get (amortized) $D=O(\log^{2}n)$, with initialization time $I=O(m\log^2n)$.
We remark that the time required to sort the matroid elements and initialize the data structure will
typically be dominated by the other terms.
It may not be immediately obvious
why it suffices to use a {\em decremental} dynamic MST data structure. The reason
for this is that we only need to be able to support edge deletions and undo operations.
The undo operation can easily be implemented to run in the same time as the edge deletion
operation by keeping track of all changes made to the data structure during the
edge deletion operations.

For the special case of MST interdiction when all interdiction costs are 1,
the previous best known running time for an enumerative algorithm is $O(n^k\log\alpha((k-1)(n+1),n))$
\cite{liang2001finding,bazgan2012efficient}.
By applying standard bounds on binomial coefficients
to the result of \cref{thm:bnbtime},
we find that for MST interdiction our algorithm runs in time
$O(\min\{(2en/k)^k,(2ek/n)^n\}\cdot D)+I$.
If we treat $n$ and $k$ as variables,
this is asymptotically faster than the previous best algorithm up to polylog factors.
Furthermore, our algorithm can handle arbitrary interdiction costs and generalizes to
arbitrary matroids.
Our algorithm also prunes many nodes using the upper bounds, and since accessing the
upper bound table takes $O(1)$ time, it does so without any increase in the asymptotic running time.
Further improvements in algorithms for the decremental dynamic MST problem would immediately
yield improvements to our algorithm as well.

The dynamic data structure we use is described formally in \cref{alg:dyninterdict}. Correctness
is easy to see: it is just an adaptation of Kruskal's algorithm.
While our dynamic data structure needs $D=O(m)$ and $I=O(m)$, and hence is theoretically slower than the data structure of Holm et al,
we found it to have excellent performance in practice, while also being considerably simpler.
This follows a common theme in the literature on dynamic MST in which simple algorithms
are repeatedly found to have better computational performance than complex
algorithms with good theoretical guarantees, for typical graphs
\cite{amato1997experimental,cattaneo2010maintaining,ribeiro2007experimental}.

Unlike other dynamic MST data structures, \cref{alg:dyninterdict}
should be easy to adapt for other matroids: only the union-find structure
needs to be replaced with a different mechanism for testing whether a set is independent in
the given matroid. As a starting point towards such an adaptation, it is easy to modify
the data structure to work for any matroid given an independence oracle. First, we
remove all references to $uf$ from \Skip, \UnSkip, and \Interdict. Then, to find the replacement edge
for $e$ in the \Interdict function, we find the lexicographically smallest basis
of $E\sm X$, remove $e$ from it, and query the independence oracle for each edge
after $e$ to see if it can be added to the basis. Depending on the specifics
of the matroid, it may be possible to improve this algorithm further.

\begin{algorithm}[p]
	\Ds{\DynInterdict}{
		\Fn{\Init{}}{
			$X\gets\emptyset$\;
			$Y\gets$ lexicographically smallest basis of $M$\;
			$e\gets$ element of $Y$ with the smallest index\;
			$uf\gets$ empty union-find structure on $n$ elements\;
			$bst\gets$ empty self-balancing binary search tree\;
			$stack\gets$ empty stack\;
		}
		\Fn{\Skip{}}{
			$uf.\Union{$e.s$, $e.t$}$\tcc*{Union the endpoints of edge $e$}
			$e\gets bst.\Next{$e$}$\tcc*{Move to the next element in $bst$ after $e$}
		}
		\Fn{\UnSkip{}}{
			$e\gets bst.\Prev{$e$}$\tcc*{Move to the element in $bst$ before $e$}
			$uf.\Undo{}$\tcc*{Undo the most recent \Union operation}
		}
		\Fn{\Interdict{}}{
			\lIf{$e=m$}{\Return false}
			\For(\tcc*[f]{Find the replacement edge for $e$}){$f=e+1,\dots,m$}{
				\lIf{$f\in Y$}{$uf.\Union{$f.s$, $f.t$}$}
				\lElseIf{$uf.\Find{$f.s$}\ne uf.\Find{$f.t$}$}{\Break}
				\lIf{$f=m$}{$f\gets\bot$}
			}
			$\Undo$ all $\Union$ operations from the above loop\;
			\lIf{$f=\bot$}{\Return false}
			$X\gets X\cup\{e\}$\;
			$Y\gets Y\sm\{e\}\cup\{f\}$\;
			$bst.\Insert{$f$}$\;
			$e\gets bst.\Erase{$e$}$\tcc*{Erase $e$ from $bst$, then replace $e$ with the next element
				after the location where $e$ was in $bst$}
			$stack.\Push{$(e,f)$}$\;
			\Return true\;
		}
		\Fn{\UnInterdict{}}{
			$e',f\gets stack.\Pop{}$\;
			$X\gets X\sm\{e'\}$\;
			$Y\gets Y\cup\{e'\}\sm\{f\}$\;
			$bst.\Erase{$f$}$\;
			$e\gets bst.\Insert{$e'$}$\tcc*{Insert $e'$ into $bst$, then set $e$ to $e'$}
		}
	}
	\caption{Dynamic interdiction data structure, specialized for MST interdiction.}
	\label{alg:dyninterdict}
\end{algorithm}

\subsection{Adaptation for the minimum cost blocker problem}
\label{sec:mincostblocker}

Recall that in the minimum cost blocker problem,
the objective is to minimize the cost required to increase the
minimum basis weight to at least $R$. In the introduction, we mentioned
that this problem is polynomial-time equivalent to the matroid interdiction
problem by performing binary search on $C$ or $R$.
In this section, we show that our branch-and-bound
scheme can in fact be modified to directly solve the minimum cost blocker
problem without incurring the overhead of performing binary search.
In \cref{sec:comp}, we compare the new algorithm (\cref{alg:bnbBlocker})
against \cref{alg:bnb} with binary search over $C$ in computational experiments, and find that \cref{alg:bnbBlocker}
is faster, as expected.

Since minimum cost blocker instances do not have a $C$ value, for the purposes of
adapting the algorithm, in the graphic matroid case we define $C=\MinCut-1$ where $\MinCut$ is
the global minimum cut of $G$ with edge costs $c$. The motivation for this is
that $\MinCut-1$ is always an upper bound on $C$ in the MST interdiction problem
(setting it any larger would allow the leader to cut the graph).
For general matroids the trivial bound $C=c(E)$ can be used.
\cref{alg:bnbBlocker} describes the modified branch-and-bound algorithm.
We omit a full proof of correctness because the proof is largely the same as the proof
of \cref{thm:bnbworks}. However, a few observations are in order.

In the edge blocker problem, a minimum $c$-cost set $X^*$ which intersects every basis is always a trivial solution,
because it causes the minimum basis weight to increase to infinity. So, we start with
such an $X^*$ as an initial incumbent solution. For graphic matroids this $X^*$ is just a global minimum cut,
but it may be NP-hard to find $X^*$ for general matroids, which we discuss later in \cref{sec:preproc}.
We then enumerate all undominated
$X$ (adopting the term {\em undominated} from the proof of \cref{thm:bnbworks}) which could potentially
lead to an improved incumbent solution $(X^*,Y^*)$. Notice that for $(X,Y)$ to be an
improved incumbent solution, $w(Y)$ must reach the target weight $R$, and the cost
of $X$ must be at most the cost of $X^*$. Furthermore, no child node of a node
with weight at least $R$ can have lesser cost, so we can prune any such node.
Similarly, no child of a node with cost at least $c(X^*)$ can have lesser cost, and since
we assume all costs are at least 1, it cannot have higher weight without also having higher cost,
so we can prune it.

The bound test $f_p(d.e, c(X^*)-c(d.X), d,X\cap\{1,\dots,p\}) + w(d.Y) < R$ works correctly because
in order to reach an improved incumbent solution, we must use
at most $c(X^*)-c(d.X)$ additional capacity, and we must reach
a solution of value at least $R$.
The background computation of the upper bound differs from the interdiction case
in that we compute $f_p(1,r,\emptyset)$ for {\em all} $r=0,1,\dots,C$.
This is necessary because the algorithm calls $f_p(e,c(X^*)-c(d.X), d.X\cap\{1,\dots,p\})$,
and $c(X^*)-c(d.X)$ is not necessarily equal to $C-c(X)$ for any $X\in\U$.

Finally, instead of interdicting the element $d.e$ whenever there is capacity to do so (as in the interdiction
problem), we instead interdict whenever a replacement element for $d.e$ exists (i.e., when
$d.\Interdict{}$ returns true).
This ensures that we enumerate all possible undominated solutions.

\begin{algorithm}[t]
	In a background thread, compute the upper bound $f_p(1,r,\phi(\emptyset),\emptyset)$ for all $r=0,1,\dots,C$
	with a progressively increasing number of prefix bits $p=0,1,\dots$ until the memory usage limit is reached\;
	Let $X^*$ be a minimum $c$-cost set which intersects every basis of $M$ and let $Y^*=\bot$\;
	$d\gets\DynInterdict.\Init{}$\;
	\Fn{\BranchAndBound{$i$}}{
		\lIf{$w(d.Y) \ge R$ and $c(d.X) \le c(X^*)$}{$X^*,Y^*\gets (d.X,d.Y)$}
		\lIf{$i=\rank(E)$ or $d.e>m$ or $w(d.Y)\ge R$ or $c(d.X) \ge c(X^*)$}{\Return}
		\If{some upper bound has finished computation}{
			Let $p=\max\{p:$ computation of $f_p$ has finished$\}$\;
			\lIf{$f_p(d.e, c(X^*)-c(d.X), d.X\cap\{1,\dots,p\}) + w(d.Y) < R$}{\Return}
		}
		\If{$d.\Interdict{}$}{
			$\BranchAndBound{$i$}$\;
			$d.\UnInterdict{}$\;
		}
		$d.\Skip{}$\;
		$\BranchAndBound{$i+1$}$\;
		$d.\UnSkip{}$\;
	}
	\BranchAndBound{$0$}\;
	\Output $(X^*,Y^*)$\;
    \caption{Modification of the branch and bound algorithm for solving the minimum cost blocker problem.}
	\label{alg:bnbBlocker}
\end{algorithm}

\subsection{Preprocessing}
\label{sec:preproc}

In this section, we discuss two preprocessing steps: one which can simplify or even solve
some matroid interdiction instances, and one which can speed up computation of the upper bound at the expense of bound quality.
The former preprocessing step has been used previously in the case of MST interdiction
\cite{zenklusen20151,linhares2017improved}. Here, we generalize the ideas to arbitrary matroids, and
describe how to perform the preprocessing in polynomial time for graphic and partition matroids.

To begin, observe that if there is some $X\in\U$ such that
$\rank(E\sm X)<\rank(E)$, then $\L(X)=\emptyset$, so the optimal value of \eqref{eq:matint} is $\infty$.
Our branch-and-bound algorithm signals an error if this
case is detected.
We would instead like to have a way of detecting this case as a preprocessing step.
The problem of finding such an $X$ is equivalent to the problem of finding a minimum cost set $X$ which
intersects every basis, known as a {\em cocircuit}.
If we have such a set $X$ which intersects every basis and $c(X)\le C$,
then $X\in\U$ and $\rank(E\sm X)<\rank(E)$.
For graphic and partition matroids, there are polynomial time algorithms
that can find a minimum cost cocircuit, which we discuss below.
However, this is NP-hard even for binary matroids (see \cite{vardy1997intractability},
which shows NP-hardness for the equivalent problem of finding the minimum distance of a
binary linear code).
So, in general we cannot do this preprocessing in polynomial time unless $\operatorname{P}=\operatorname{NP}$.

For graphic matroids, we can easily find a minimum cost cocircuit in polynomial time: such a set is simply
a global minimum cut in the graph.
Now we consider partition matroids. Suppose that the partition matroid $M=(E,\I)$ is represented by
$\I=\{J\se E:|J\cap S_i|\le k_i~\forall i\}$ for some disjoint sets $S_i$ and integers $k_i$. Then, a cocircuit is any set $X$ which
has $|X\cap S_i|>|S_i|-k_i$ for some $i$, because any basis $B$
has $|B\cap S_i|=k_i$, so $|X\cap S_i|+|B\cap S_i|>|S_i|$ and hence $X\cap B\ne\emptyset$.
Since $c_e\ge1~\forall e\in E$, any minimum cost cocircuit must only intersect a single set $S_i$,
and hence we can find such a cocircuit by using a greedy algorithm to find the
cheapest $i$ and $X$ such that $|X\cap S_i|=|S_i|-k_i+1$.

An extension of this idea can be used to simplify some problem instances, as observed in prior works
\cite{zenklusen20151,linhares2017improved}.
Let $w'_1,w'_2,\dots,w'_\ell$ be the distinct weights from $w$, and assume
$w'_1<w'_2<\dots<w'_\ell$.
For all $i$, let $E_{\le i}=\{e\in E:w_e\le w'_i\}$.
Let $k$ be the smallest index such that $\rank(E_{\le k}\sm X)=\rank(E)$
for all $X\in\U$. Then, we claim that for any $X\in\U$,
$\bmin(X\cap E_{\le k-1})=\bmin(X)$.
The following lemma is a trivial generalization of Lemma 2.2 from \cite{linhares2017improved},
which concerns only graphic matroids. The proof is almost identical, but we include it for completeness.
To derive it from the original
proof, just substitute $\rank(E)-\rank(S)+1$ for every occurrence of $\sigma(S)$, where $S$ is an arbitrary subset of $E$
(here, $\sigma(S)$ is just the number of connected components of the graph $(V,S)$). 
\begin{lemma}
For convenience, define $w'_0:=0$.
If $X\in\U$, then
\[\min\{w(Y):Y\in\L(X)\}=\sum_{i=0}^{k-1}(w'_{i+1}-w'_i)(\rank(E)-\rank(E_{\le i}\sm X)).\]
\end{lemma}
\begin{proof}
Consider running the greedy algorithm to obtain a minimum weight basis of $M\sm X$.
For every $1\le j\le\ell$, the greedy algorithm includes
$\rank(E_{\le j}\sm X)-\rank(E_{\le j-1}\sm X)$ elements of weight $w'_j$.
Note that $\rank(E_{\le j}\sm X)-\rank(E_{\le j-1}\sm X)=0$ for $j>k$ because the rank of both
sets is $\rank(E)$ by the definition of $k$. Now, we have
\begin{align*}
\min\{w(Y):Y\in\L(X)\} &= \sum_{j=1}^{\ell}w'_j(\rank(E_{\le j}\sm X)-\rank(E_{\le j-1}\sm X))
	= \sum_{j=1}^{k}w'_j(\rank(E_{\le j}\sm X)-\rank(E_{\le j-1}\sm X))\\
	&= \sum_{i=0}^{j-1}(w'_{i+1}-w'_i)\sum_{j=1}^{k}(\rank(E_{\le j}\sm X)-\rank(E_{\le j-1}\sm X))\\
	&= \sum_{i=0}^{k-1}(w'_{i+1}-w'_i)\sum_{j=i+1}^k(\rank(E_{\le j}\sm X)-\rank(E_{\le j-1}\sm X))\\
	&= \sum_{i=0}^{k-1}(w'_{i+1}-w'_i)(\rank(E_{\le k}\sm X)-\rank(E_{\le i}\sm X))
	 = \sum_{i=0}^{k-1}(w'_{i+1}-w'_i)(\rank(E)-\rank(E_{\le i}\sm X)).
		\tag*{\qedhere}
\end{align*}
\end{proof}
As a consequence, it suffices to consider only interdiction sets $X$ such that
$X\se E_{\le k-1}$, because the objective value of the lower level problem has no dependence
on whether elements in $E_{\ge k}$ are interdicted.
In other words, if we define
\begin{equation*}
\bar m=\max\{e:e\in E_{\le k-1}\},
\end{equation*}
then we can assume every $e\in X$ has $e\le\bar m$.
Furthermore, we can delete all elements of weight larger than $w'_k$, because there
is always a (minimum weight) basis using only elements in $E_{\le k}\sm X$.
For matroids where we can efficiently find a minimum cost cocircuit,
we compute $\bar m$ by performing binary search to find the smallest $i$
such that minimum cost of a cocircuit in $E_{\le i}$ is greater than $C$.
\cref{alg:bnb} can be modified to benefit from this observation by replacing the
condition $d.e>m$ with $d.e>\bar m$ on \cref{algline:bnbbasecase}.
\cref{alg:bnbBlocker} can be modified similarly (using the assumption that $C=\MinCut-1$ for
the purposes of computing $\bar m$).

We now discuss the second preprocessing step.
Memory usage requirements can be an issue
for computing the upper bound $f$ (as described in \cref{sec:ub}), especially if $C$ is large. To address this we describe
a way to round the costs and capacities used in the upper bound computation while ensuring that it remains an upper bound.
This reduces the memory and time required at the expense of the quality of the bounds.
When $C$ is particularly large, performing this rounding
step can be beneficial to the overall running time.

The way we accomplish this is as follows. We pick some constant $K$, replace
all interdiction costs $c_e$ by $\lfloor c_e/K\rfloor$, and replace the capacity $C$ by $\lceil C/K\rceil$.
Since we round down the costs and round up the capacity, all feasible upper-level solutions
remain feasible, but some new upper-level solutions may become feasible which were not previously.
The result is a relaxation of the original feasible region, and since we are dealing
with a maximization problem, this produces an upper bound. In practice, we only perform this rounding
on a group of instances from the literature where the costs are on the order of $10^8$,
and hence it would be completely impractical to use the DP algorithm without rounding.
The details can be found in \cref{sec:compPrior}.

\section{Computational results}
\label{sec:comp}

We evaluated our algorithm computationally for the MST interdiction problem and
the minimum cost MST edge blocker problem.
We compare our results to two previous papers which also published computational results.
Overall, our algorithm performs exceedingly well compared to past algorithms.
For many instances, we improved the best run time from hours to seconds.
This was enough to solve all previously unsolved instances from the literature.

We implemented our algorithm in C++, compiled it using the Clang 17 compiler,
and tested it on an Intel Gold 6148 Skylake CPU @ 2.4 GHz.
Our solver was limited to use at most 32 GB of RAM.
Our code and instances are open source
and are available at
\url{https://github.com/nwoeanhinnogaehr/mstisolver}.
The correctness of our implementation was verified by comparing the output against
algorithms such as the MIP extended formulation \cite{bazgan2012efficient, wei2021integer}
or a basic brute-force algorithm
whenever it was possible for these slower algorithms to solve the instances.

This section is divided into two parts. In the first, we compare the performance
of our algorithms to previously published algorithms. The second examines in detail
how various solver features and instance parameters affect the performance of our algorithms.

\subsection{Comparison with prior works}
\label{sec:compPrior}

We compare our algorithm directly to two previous works: the (mixed) integer programming
formulations for the minimum cost MST edge blocker problem by Wei, Walteros, and Pajouh \cite{wei2021integer},
and the branch-and-bound algorithms for the special
case of MST interdiction where all interdiction costs are 1
by Bazgan, Toubaline, and Vanderpooten \cite{bazgan2012efficient}.
As the original implementations were not available, the comparison is made directly
to the running times reported in the original papers.
Despite this, we believe the gap in performance is significant enough to
rule out the possibility of the improvements coming from compute speed or implementation quality alone.
Note that with default settings, our solver uses two threads: one for branch-and-bound, and one for computing upper bounds,
so the CPU time is roughly double the wall-clock time.
The reason why we compute upper bounds on a separate thread is that it is difficult to determine
ahead of time how difficult an instance is and thus how strong of an upper bound we should invest in computing.
By doing it in parallel, we ensure that strong bounds are eventually computed for hard instances and that
time is not wasted on strong bounds for easy instances.
In \cref{sec:perfDetails}, we show that using two threads in this way does not offer any significant advantage over
periodically switching between branch-and-bound search and upper bound computation on a single thread,
and hence it is fair to compare our algorithm to other serial algorithms.
Since it is unclear if the
previous papers used a parallelized implementation or whether they measured CPU time or wall-clock time,
we report all running times as CPU time to ensure as fair a comparison as possible.

The minimum cost MST edge blocker tests performed by Wei, Walteros, and Pajouh were run on a
12-core Intel Xeon E5-2620 v3 CPU at 2.4GHz with 128 GB of RAM, using C++ and Gurobi 8 \cite{wei2021integer}.
We tested our algorithm on the instances available for download in the online supplementary material for the paper.
For these instances, the authors used a parameter $\gamma\in\{0.05,0.3,0.7,0.95\}$ to set the value of $R$ as a function
of the minimum and maximum spanning tree weight. Specifically, they set $R=(\overline w-\underline w)\gamma+\underline w$
where $\overline w$ is the maximum spanning tree weight and $\underline w$ is the minimum
spanning tree weight. In our tables we only report results for $\gamma=0.05$,
because we found that for $\gamma\in\{0.3,0.7,0.95\}$, the optimal solution for all instances 
is for the leader to use a minimum cut. Our algorithm is able to very quickly
handle this trivial case, solving all such instances in just a few milliseconds.
The instances are grouped by number of vertices $n$ and approximate number of edges $\tilde m$.
Exactly how $\tilde m$ is determined is described in the original paper \cite{wei2021integer}.
We remark that the instances available in the supplementary material do not exactly
match what is described in the paper. Specifically, the data appears to include
20 instances for each value of $n$ and $\tilde m$, but the paper only references 12.
We tested all 20, as we found no indication of which subset of these instances were originally
tested. We also adjusted the weights and costs of the downloaded instances to fall in the range $[20,80]$, as reported in the paper. In any case, this adjustment had no measurable impact on performance.

Since these instances have weights and costs specified with 6 digits of decimal precision
but our implementation only supports integer weights and costs, we scaled up all values
by $10^6$ so that they became integers. However, costs of this magnitude are not suitable
for use with our pseudopolynomial time DP bounds, so to handle this we rounded the costs
and capacities used in the DP bound as described in \cref{sec:preproc}, with $K=10^5$.
This reduced the costs to integers less than 1000,
which are of an appropriate magnitude for the DP algorithm. We emphasize that this
does not affect the accuracy of the solution output because
only the bounds are weakened; the costs and capacity used in the branching part of the algorithm are not rounded.

\begin{table}
    \centering
	\small
	\begin{minipage}{0.52\textwidth}
	\centering
    \begin{filecontents*}{data/wei-table.csv}
n,m,AvgTime,AvgTime,WeiAvgTime,WeiPercOpt
40,78,0.00359335,0.017665950000000003,0.04000000000000001,100.0
40,118,0.007950200000000001,0.0465043,0.2599999999999999,100.0
80,160,0.01418085,0.0819488,0.4900000000000001,100.0
80,240,0.06318560000000001,0.4333134,1943.2300000000002,75.0
80,305,0.14732779999999995,0.5465863499999999,1469.4900000000007,83.0
160,318,0.051259650000000004,0.11495015,97.27999999999999,100.0
160,476,0.11497759999999999,0.30791694999999997,3829.9499999999985,48.0
160,624,0.20554424999999998,0.9029306499999998,6066.429999999998,20.0
200,398,0.07517735,0.13555694999999995,280.7600000000001,100.0
200,596,0.17564715000000003,0.22622905,3757.7399999999993,48.0
200,784,0.3267953,0.7612481000000001,7200.0,0.0
\end{filecontents*}
\pgfplotstabletypeset[
        col sep=comma, 
        string type,
        column type=l,
        every head row/.style={
            before row={
                \multicolumn{2}{c|}{}&\multicolumn{1}{c|}{BnB} & \multicolumn{1}{c|}{BS-BnB} & \multicolumn{2}{c}{Best from \cite{wei2021integer}}\\
                \bottomrule
            },
        },
        precision=3,
        fixed,
        display columns/0/.style={column type=l, string type, column name=$n$},
        display columns/1/.style={string type, column name=$\tilde m$},
        display columns/2/.style={column type=|l, numeric type, column name=Time},
        display columns/3/.style={column type=|l, numeric type, column name=Time},
        display columns/4/.style={column type=|l, numeric type, column name=Time},
        display columns/5/.style={column type=l, numeric type, column name=Opt\%},
		every row no 0/.style={before row=\hline},
		every row no 2/.style={before row=\hline},
		every row no 5/.style={before row=\hline},
		every row no 8/.style={before row=\hline},
        ]{data/wei-table.csv}
    \caption{Optimality percentage and average running times for instances by Wei, Walteros, and Pajouh \cite{wei2021integer}.
	Our algorithms (BNB and BS-BnB) solved all instances to optimality.}
	\label{tab:wei}
	\end{minipage}\qquad
	\begin{minipage}{0.42\textwidth}
	\centering
    \begin{filecontents*}{data/bazgan-table.csv}
n,k,AvgTime,BazganAvgTime
20,3,0.0035654199999999993,0.0
20,5,0.012264679999999997,0.03200000000000002
20,7,0.03579766,0.3800000000000001
20,9,0.12662176000000003,3.0470000000000024
25,3,0.00498737,0.0
25,5,0.019182880000000003,0.08900000000000022
25,7,0.07097500000000001,1.6169999999999993
25,8,0.21405391000000001,3.565999999999994
30,3,0.006755920000000001,0.0
30,5,0.02032675999999999,0.23100000000000043
30,7,0.14269720999999996,3.5529999999999977
50,3,0.028996950000000004,0.02599999999999997
50,5,0.10190119999999996,1.8559999999999968
50,7,1.923742859999999,81.70700000000012
75,3,0.05497128000000002,0.09600000000000007
75,5,0.24547750000000002,10.458999999999977
75,7,9.06209855,650.0080000000012
100,3,0.10800383000000008,0.2100000000000005
100,5,0.5755558100000001,49.89500000000007
100,7,21.511955449999995,2016.4100000000028
200,5,8.057157170000002,572.5570000000006
300,5,57.5122573,1793.46
400,5,117.5560615,7265.849999999988
\end{filecontents*}
\pgfplotstabletypeset[
        col sep=comma, 
        string type,
        column type=l,
        every head row/.style={
            before row={
                \multicolumn{2}{c}{} \vrule & \multicolumn{1}{c}{BnB} \vrule & \multicolumn{1}{c}{Best from \cite{bazgan2012efficient}}\\
				\bottomrule
            },
            after row=\hline,
        },
        precision=3,
        fixed,
        display columns/0/.style={string type, column name=$n$},
        display columns/1/.style={string type, column name=$C$},
        display columns/2/.style={column type=|l, numeric type, column name=Time},
        display columns/3/.style={column type=|l, numeric type, column name=Time},
		every row no 4/.style={before row=\hline},
		every row no 8/.style={before row=\hline},
		every row no 11/.style={before row=\hline},
		every row no 14/.style={before row=\hline},
		every row no 17/.style={before row=\hline},
		every row no 20/.style={before row=\hline},
		every row no 21/.style={before row=\hline},
		every row no 22/.style={before row=\hline},
        ]{data/bazgan-table.csv}
    \caption{Average running times for instances by Bazgan, Toubaline, and Vanderpooten \cite{bazgan2012efficient}.
	Both approaches solved all instances to optimality.}
    \label{tab:bazgan}
	\end{minipage}    
\end{table}

We report our results on these instances in \cref{tab:wei}.
In this table, each row corresponds to the instances with a specific value of $n$ and $\tilde m$.
Under the headings `BnB' and `BS-BnB' we report the average CPU time in seconds for \cref{alg:bnbBlocker}
and binary search over the value of $C$ in \cref{alg:bnb}, respectively.
Both of our algorithms solved all instances to optimality.
Under the heading `Best from \cite{wei2021integer}' we include the best result from any algorithm
described in the paper by Wei, Walteros, and Pajouh. The original data can be found in
Table 2 of their paper. The `Time' column is the best average running time in seconds
for any of the three algorithms EXT, CST-1, or CES (or 7200 seconds if the time-out is reached), and the `Opt\%' column
is the best percentage of instances solved to optimality.
Excluded from the table are the real-world instances, for which the optimal solution to all instances is
the minimum cut.
The real-world instances are solved within a few seconds by our algorithms, which is slower
than the time reported in \cite{wei2021integer}; the reason is that determining $\bar m$ for these instances
is slow because the graphs are very large. If computation of $\bar m$ is disabled, they are solved almost instantly.

The MST interdiction tests by Bazgan, Toubaline, and Vanderpooten were performed on a 3.4GHz processor with 3 GB of RAM, using
the C programming language \cite{bazgan2012efficient}.
The exact instances they used were not available for download, but the method used to generate the
instances was described precisely, so we were able to recreate the data set (up to the choice of random values).
Specifically, all graphs are complete graphs on $n$ vertices, the weights are integers chosen uniformly
at random in the range $[0,100]$, and all interdiction costs are 1.
For each choice of $n$ and interdiction budget $C$, the original paper generated 10 instances.
To reduce any bias from the random choice of weights on the average running time we instead
generated 100 instances for each choice of parameters.

We report our results on these instances in \cref{tab:bazgan}.
In this table, each row corresponds to the instances with a specific value of $n$ and $C$.
The columns under the heading `BnB' describe the results for \cref{alg:bnb}.
Under the `Best from \cite{bazgan2012efficient}' heading, we include the best result reported from the original paper,
across each of their four algorithms. The original data for these columns
can be found in Table 2 of their paper. The meaning of the individual columns
is the same as in \cref{tab:wei}. Note that our algorithm solved all instances
to optimality, as did the best algorithm from the original paper.

In both \cref{tab:wei,tab:bazgan}, our branch and bound algorithms can clearly be seen to offer a substantial
improvement in running time over all previous mixed integer programming and branch and bound algorithms.

\subsection{Detailed performance evaluation}
\label{sec:perfDetails}

While the previous section has established the competitive performance of our algorithm,
the instances which appeared in the literature do not test some important cases. For example,
all of the minimum cost MST edge blocker instances are relatively sparse, and all of the
MST interdiction instances are complete graphs with unit costs and a very small interdiction budget.
In this section we test some newly generated instances to better understand how the performance
of our algorithm depends on various qualities of an instance. We also perform a feature knockout
test to determine which algorithm features are most important for good performance.

We chose to focus on generating MST interdiction instances which are small (in terms of $n$ and $m$)
but which take a long time to solve using current solvers. In fact, many of these new instances with only 20 vertices take
significantly longer to solve than instances from the literature with 400 vertices.
The 1024 new instances were generated as follows. For each choice of $n\in\{10,15,20,25\}$,
$\gamma\in\{0.25,0.5,0.75,1\}$, $d\in\{0.5,0.66,0.83,1\}$, $c_{\max}\in\{1,100,1000,10000\}$,
and $w_{\max}\in\{2,100,10000,1000000\}$, we generated one instance.
The instance has $n$ vertices, $m=\lfloor d\binom n2\rfloor$ edges (chosen uniformly at random), integer costs $c$ uniformly
distributed in $[1,c_{\max}]$, integer weights $w$ uniformly distributed in $[1,w_{\max}]$
and capacity $\lfloor \gamma(M-1)\rfloor$ where $M$ is the cost of the global minimum cut for costs $c$.
Note that $M-1$ is the largest nontrivial capacity, as any larger would produce an
instance in which the leader is able to cut the graph.

In \cref{tab:newInstances}, we report the results for these new instances.
The rows of the table are split into five sections, each with four rows;
in each section the instances are grouped based on the value of one of
the five instance parameters $n$, $\gamma$, $d$, $c_{\max}$, or $w_{\max}$.
That is, each section reports results for all instances, but instances
are aggregated differently in different sections.

The first four data columns, `Opt\%', `Time', `MaxTime'
indicate the percentage of instances solved to optimality, average CPU time, and maximum CPU time, respectively.
All instances were tested with a 1 hour (3600 second) CPU time limit. For instances where the limit is reached,
the time limit is used in place of the running time.
The next four columns describe the quality of bounds: `LB' and `MaxLB' are the average and maximum percent gap between
the greedy lower bound and the optimal solution, and `UB' and `MaxUB' are the average and maximum percent gap between
the (strengthened) upper bound and the optimal solution. The gap percentage is computed as $100\cdot|z^*-z|/z^*$ where
$z$ is the approximate solution weight (lower or upper bound) and $z^*$ is the optimal solution weight.
The `UBTime' column reports the average number of seconds required to compute the upper bound $f(1,C,\phi(\emptyset))$.
For any $p$, the time required to compute $f_p(1,C,\phi(\emptyset),\emptyset)$ can be accurately
estimated by multiplying this value by $2^p$.
Only instances which were solved to optimality are included in the average for the last four columns.

The results indicate that the number of vertices $n$ is the best predictor of instance difficulty:
while instances with $n=10$ are all solved within 20 milliseconds, around 7\% of the instances with $n=25$ could not
be solved within 1 hour.
However, to be hard an instance must have high density $d$ and high capacity factor $\gamma$.
This matches the behaviour seen in instances from the literature.
The instances in \cref{tab:wei} are solved very quickly, and the graphs are very sparse (i.e., $d$ is small).
On the other hand, if $\gamma$ is small, as is the case in \cref{tab:bazgan}, instances are still largely tractable as well.
The parameters $c_{\max}$ and $w_{\max}$ have less impact on the performance.
Instances with smaller $c_{\max}$ are typically solved faster because the DP table is smaller and can
be computed more quickly. On the other hand, instances with smaller $w_{\max}$ are somewhat harder, but only
by a small factor; we are not sure why this is the case.
We remark that the number of branch-and-bound nodes is often on the order of $10^9$ for the hardest instances.
This is expected since our algorithm only needs $O(m)$ time per node.

Surprisingly, we found that across all instances,
the greedy lower bound is 1.48\% from optimal on average, and is never
worse than 17.82\% from optimal. These statistics hold for instances from the literature as well.
In \cref{sec:greedy} we saw that the worst case can be arbitrarily bad even for
a graph on four vertices, but evidently the performance on random graphs is much better.
Therefore, this greedy algorithm alone could serve as a very strong heuristic in applications where
an exact algorithm is not needed.

As for the strengthened DP upper bounds, they have an average optimality gap of 14.76\%
at the root node, and are never worse than 68.24\%. The parameter which has the largest correlation with
the time required to compute the bounds is $c_{\max}$, because the running time is directly a function of it.
The running time also depends on $m$, but as the magnitude of $m$ is much smaller, it has less of an impact.
While there is no clear connection
between $n$ and the bound strength, all other parameters do appear to be correlated with the bound strength.
In \cref{sec:mstworstcase} we showed
an example with an optimality gap linear in $n$, but evidently this case is not commonly seen in random graphs.
This offers some indication as to why the branch-and-bound algorithm performs so well:
although it is possible for the upper bounds to be weak, the bounds are used across a large number
of subproblems. Due to the good average-case performance, it appears to be quite challenging to
construct an instance where the bounds are consistently weak for all subproblems. The hardest
instances we currently know of have $d=1$, $\gamma=1$,
and very large $c_{\max}$. However, these instances are still far from the
theoretical worst case.

\begin{table}
    \centering
	\small
    \begin{filecontents*}{data/new-all.csv}
$n$,nOpt,AvgTime,MaxTime,LBGap,MaxLBGap,UBGap,MaxUBGap,UBTime
10,100.0,0.0027223671874999993,0.018147,1.1093043302383816,16.002690397350992,9.093129191963698,44.750604948872066,0.000295726287739071
15,100.0,0.13084621875,10.0402,1.6081518800050634,15.879828326180258,16.780777238532337,68.23529411764706,0.011601652698312714
20,99.609375,70.31410414453126,3600,1.5248997227416579,17.249121436938697,16.754303473587616,54.960520664706465,0.03878964780894866
25,92.96875,380.9927031367188,3600,1.67217418069829,17.816935878609886,16.401176776744524,45.946949573198914,0.06664460323599601
$\gamma$,,,,,,,,
0.25,100.0,0.0080213046875,0.322392,0.8442056141521626,10.0,4.9408754370742605,25.713483146067414,0.0038011939685289944
0.5,100.0,4.49297728515625,266.425,1.2811397010754577,15.879828326180258,13.139565725665554,45.756218905472636,0.027545892751138373
0.75,98.4375,116.24739396484378,3600,1.3638212516640273,17.249121436938697,18.050188163105535,60.26858734410758,0.046942893314596036
1.0,94.140625,330.6919833125,3600,2.4677357593777254,17.816935878609886,20.96944158751351,68.23529411764706,0.05539682022636778
$d$,,,,,,,,
0.5,100.0,0.11572203515624999,8.70384,0.914267283888697,15.921787709497206,9.300965581311651,43.47826086956522,0.013761229352406553
0.66,100.0,5.66507005859375,381.702,1.108587948027227,14.150883622915211,13.574084491178262,45.946949573198914,0.029062165746986154
0.83,99.21875,77.83608011718749,3600,1.7765963934730788,17.249121436938697,16.822354462114458,46.007748550787966,0.04501997847965295
1.0,93.359375,367.82350365625,3600,2.148072179142254,17.816935878609886,20.034668053309243,68.23529411764706,0.05228685936314698
$c_{\max}$,,,,,,,,
1,98.828125,69.47482009765626,3600,1.3420503112258229,17.816935878609886,9.352734515274705,29.918032786885245,0.0018290014469433016
100,98.828125,106.19779596874993,3600,1.5211373900742224,16.002690397350992,15.384724431546562,45.756218905472636,0.0029417401727547523
1000,97.65625,129.97889789453131,3600,1.5849857024588299,17.249121436938697,19.358547352459723,60.26858734410758,0.018321786913847997
10000,97.265625,145.78886190625002,3600,1.4532640967229655,11.396610169491526,24.590830089640164,68.23529411764706,0.2227051971448182
$w_{\max}$,,,,,,,,
2,96.484375,177.8450875781251,3600,0.6272727268513013,10.0,8.876622934446697,31.25,0.01785036614963556
100,98.4375,116.14311912499994,3600,1.778444477544068,15.921787709497206,17.665441039022173,68.23529411764706,0.030056848244560716
10000,99.21875,68.68216094921878,3600,1.8603956381426143,17.816935878609886,18.786873655385527,54.960520664706465,0.05925024814733861
1000000,98.4375,88.77000821484371,3600,1.614487341928717,14.150883622915211,19.365319348331997,60.26858734410758,0.05075848817995984
\end{filecontents*}
\pgfplotstabletypeset[
        col sep=comma, 
        string type,
        column type=l,
        every head row/.style={
            before row={
            },
            after row=\hline,
        },
        precision=2,
        fixed,
        display columns/0/.style={string type},
        display columns/1/.style={column type=|l, numeric type, column name=Opt\%},
        display columns/2/.style={numeric type, column name=Time},
        display columns/3/.style={numeric type, column name=MaxTime},
        display columns/4/.style={numeric type, column name=LB},
        display columns/5/.style={numeric type, column name=MaxLB},
        display columns/6/.style={numeric type, column name=UB},
        display columns/7/.style={numeric type, column name=MaxUB},
        display columns/8/.style={precision=4, numeric type, column name=UBTime},
		every row no 3/.style={before row=\vspace{1mm}},
		every row no 4/.style={after row=\hline},
		every row no 8/.style={before row=\vspace{1mm}},
		every row no 9/.style={after row=\hline},
		every row no 13/.style={before row=\vspace{1mm}},
		every row no 14/.style={after row=\hline},
		every row no 18/.style={before row=\vspace{1mm}},
		every row no 19/.style={after row=\hline},
        ]{data/new-all.csv}
    \caption{Aggregated statistics for new instances.}
	\label{tab:newInstances}
\end{table}

\begin{table}
    \centering
	\small
    \begin{filecontents*}{data/knockout-table.csv}
n,nOpt,Time,nOpt,Time,nOpt,Time,nOpt,Time,nOpt,Time
10,100.0,0.0027223671874999993,100.0,0.00178860546875,100.0,0.0025745820312499997,100.0,0.0022184765625,100.0,0.0016737734374999995
15,100.0,0.13084621875,100.0,0.24288878125000005,100.0,0.10458976953125,100.0,0.1305894296875,100.0,0.15458714453125003
20,99.609375,70.31410414453126,97.65625,166.9326841953125,99.609375,77.06405095703124,100.0,68.09276818359375,100.0,74.86314165625001
25,92.96875,380.9927031367188,84.31372549019608,680.481432698039,88.671875,515.1531481328126,92.578125,383.8125295585937,91.796875,352.55475694921876
$\gamma$,,,,,,,,,,
0.25,100.0,0.0080213046875,100.0,0.014153746093750003,100.0,0.006919714843749999,100.0,0.00701787109375,100.0,0.010803457031249998
0.5,100.0,4.49297728515625,100.0,12.625712046874998,100.0,6.63761179296875,100.0,4.5435582734375,100.0,4.6251970273437495
0.75,98.4375,116.24739396484378,93.75,260.8562402109375,96.875,167.93907155468753,98.046875,118.12271715234371,97.65625,104.67081878125002
1.0,94.140625,330.6919833125,88.23529411764706,573.7457520235295,91.40625,417.74076037890626,94.53125,329.3648123515625,94.140625,318.2673402578125
$d$,,,,,,,,,,
0.5,100.0,0.11572203515624999,100.0,0.2845631875,100.0,0.067699921875,100.0,0.11076686328125,100.0,0.11106134375000001
0.66,100.0,5.66507005859375,100.0,24.74907994921875,100.0,6.7712612109375,100.0,5.545363816406249,100.0,5.6887989453125
0.83,99.21875,77.83608011718749,96.09375,219.21143739843748,97.65625,130.55787567187497,99.21875,79.25891642968749,98.4375,79.00236169921872
1.0,93.359375,367.82350365625,85.88235294117646,603.1114873960785,90.625,454.92752663671877,93.359375,367.1230585390625,93.359375,342.77193753515627
$c_{\max}$,,,,,,,,,,
1,98.828125,69.47482009765626,94.140625,234.2815731054688,98.046875,90.65481540234376,98.828125,69.69975794921876,98.828125,72.21543778124996
100,98.828125,106.19779596874993,95.703125,215.6848046835937,97.265625,165.8854774140625,98.828125,107.2590917773437,98.046875,99.83298638281252
1000,97.65625,129.97889789453131,96.09375,200.52302945312488,96.875,168.02131193359375,97.265625,131.07670419140624,97.265625,120.8370311640625
10000,97.265625,145.78886190625002,96.07843137254902,195.27404568235298,96.09375,167.76275869140625,97.65625,144.00255173046872,97.65625,134.68870419531248
$w_{\max}$,,,,,,,,,,
2,96.484375,177.8450875781251,96.47058823529412,178.71507194901955,96.484375,171.2849247070314,96.484375,177.7548864257813,96.875,174.6056074257812
100,98.4375,116.14311912499994,94.921875,231.17108027734378,96.484375,155.68612415625,98.4375,116.66968845703117,97.65625,112.21346729296874
10000,99.21875,68.68216094921878,94.921875,226.9258614570312,97.265625,133.75279115624997,99.21875,68.82653044531253,99.21875,60.932372011718755
1000000,98.4375,88.77000821484371,95.703125,208.88675574999994,98.046875,131.60052342187493,98.4375,88.78700032031249,98.046875,79.82271279296874
\end{filecontents*}
\pgfplotstabletypeset[
        col sep=comma, 
        string type,
        column type=l,
        every head row/.style={
            before row={&\multicolumn{2}{c|}{All features} & \multicolumn{2}{c|}{No UB} & \multicolumn{2}{c|}{Weak UB}& \multicolumn{2}{c|}{No LB} & \multicolumn{2}{c}{Single thread}
				\vspace{1mm}\\
				\hline
            },
            after row=\hline,
        },
        precision=2,
        fixed,
        display columns/0/.style={column type=l|, string type},
        display columns/1/.style={column type=l, numeric type, column name=Opt\%},
        display columns/2/.style={numeric type, column name=Time},
        display columns/3/.style={column type=|l, numeric type, column name=Opt\%},
        display columns/4/.style={numeric type, column name=Time},
        display columns/5/.style={column type=|l, numeric type, column name=Opt\%},
        display columns/6/.style={numeric type, column name=Time},
        display columns/7/.style={column type=|l, numeric type, column name=Opt\%},
        display columns/8/.style={numeric type, column name=Time},
        display columns/9/.style={column type=|l, numeric type, column name=Opt\%},
        display columns/10/.style={numeric type, column name=Time},
		every row no 3/.style={before row=\vspace{1mm}},
		every row no 4/.style={after row=\hline},
		every row no 8/.style={before row=\vspace{1mm}},
		every row no 9/.style={after row=\hline},
		every row no 13/.style={before row=\vspace{1mm}},
		every row no 14/.style={after row=\hline},
		every row no 18/.style={before row=\vspace{1mm}},
		every row no 19/.style={after row=\hline},
        ]{data/knockout-table.csv}
    \caption{Knockout test.}
	\label{tab:knockout}
\end{table}

We now examine the impact of different solver features on performance in a knockout test.
The features considered in the test are: (1) the DP upper bound,
(2) the strengthened DP upper bound,
(3) the greedy lower bound, and (4) the computation of branch-and-bound search and upper bounds in parallel.
For these tests we again used a 3600 CPU second time limit.
The results are summarized in \cref{tab:knockout}.
As usual, the `Opt\%' column is the percentage of instances solved to optimality,
and the `Time' column is the average CPU time in seconds. Instances
which reached the time limit are counted as having taken 3600 seconds.
The aggregation of instances into sections is the same as in \cref{tab:newInstances}.

It can easily been seen from the table that disabling the upper bound entirely
or disabling the strengthened upper bound causes a clear reduction in performance,
especially on more difficult instances.
This is true even considering that without the upper bound test, our branch-and-bound
implementation is able to enumerate nodes about two times faster.
On the other hand, disabling the greedy lower bound has no clear impact
on the optimality percentage and the running time.
This may be surprising considering how we found in \cref{tab:newInstances}
that the greedy lower bound has exceptionally good performance. The lack of any
significant impact on the exact algorithm suggests that typically the
branch-and-bound scheme very quickly finds solutions of quality
comparable to the greedy lower bound.

For the single-threaded variant of the solver (feature column `Single thread'), the implementation
specifically does the following. A timer is set initially to expire after $0.01$ seconds.
We then repeat the following until the problem is solved.
First, branch-and-bound search begins. When the timer expires, it is reset and the solver switches to
performing upper bound computation. When the timer expires again, the upper bound computation
is interrupted, the length of the timer is doubled, and the process repeats (i.e., branch-and-bound resumes,
but with double the amount of time on the timer). This method of alternating between search and bounds
ensures easy instances are solved quickly by branching (without spending a lot of time on unnecessary bound computations)
but for hard instances we spend sufficient time on improving the bounds.
As can be seen from the table, this change to the solver only has a minimal impact on the performance
and hence it is fair to compare our (technically multi-threaded) algorithm to other single-threaded algorithms,
especially considering the wide gaps in performance seen in \cref{sec:compPrior}.

\section{The matroid inclusion-interdiction problem}
\label{sec:dual}

Typically, interdiction problems consider an adversary who is excluding some structure from being used
in the solution to the lower level problem. However, besides being motivated by applications, the choice
of excluding rather then {\em including} is perhaps arbitrary. In this section we consider this dual
problem in which the leader forces certain elements to be included in the follower's basis, and show
that the two problems are connected through matroid duality.

Let $M=(E,\I)$ be a matroid and let $M^*=(E,\I^*)$ be its dual. Let $\B$ be the set of bases
of $M$, and let $\B^*$ be the set of bases of $M^*$. Let $w\in\Z^E$.
By the definition of the dual matroid, we know that $Y\in\B^*$ if and only if $E\sm Y\in\B$.
This immediately yields the following.
\begin{proposition}\label{prop:inclusion}\begin{align}
&\max\{\min\{w(Y):Y\in\B^*,Y\se E\sm X\}:X\in\U\}\label{eqn:dualA}\\
=w(E)-&\min\{\max\{w(Y):Y\in\B,X\se Y\}:X\in\U\}.\label{eqn:dualB}
\end{align}
\end{proposition}
This proposition connects these two types of interdiction problems
on matroids: \labelcref{eqn:dualA}, in which the leader can force elements to be excluded from a basis,
and \labelcref{eqn:dualB}, in which in which the leader
can force elements to be included in a basis. We will call this inclusion variant
of the problem the {\em matroid inclusion-interdiction problem}.

Since our exact algorithms do not require non-negative weights, \cref{prop:inclusion} implies that
we can solve the matroid inclusion-interdiction problem on a matroid $M$ whenever
we can solve the matroid interdiction problem on the dual of $M$. The same holds for the minimum cost
blocker variants.
The following is now immediate from the fact that both partition matroids and
graphic matroids from planar graphs are self-dual.
\begin{corollary}
	There is a pseudopolynomial time algorithm for the partition matroid inclusion-interdiction problem.
\end{corollary}
\begin{corollary}
	\cref{alg:bnb} (with bounds computed by \cref{alg:computeGFnImproved}) can be used
	to solve the planar graphic matroid inclusion-interdiction problem.
\end{corollary}
We do not know how to apply our techniques from \cref{sec:ub} to get bounds
when $M$ is the bond matroid of a non-planar graph, so it remains open to derive good bounds
for matroid inclusion-interdiction on non-planar graphic matroids or in even more general settings.

\section{Conclusion}

We have presented new combinatorial algorithms and theoretical results for matroid interdiction and related problems.
In the special case of MST interdiction, our algorithms achieve state-of-the-art computational performance,
improving on the previous best results by a few orders of magnitude.

There are a number of interesting directions for future work.
It is natural to consider the generalization of these methods to the min-max problem where the follower wants to find
a maximum weight {\em independent set} (as opposed to a maximum weight basis, as has been the focus of this paper);
this would provide a connection with the rank reduction problem \cite{joret2015reducing}.
We hope that our bound framework will be applied to other matroids (some obvious next targets are laminar matroids, bond matroids, 
transversal matroids, and gammoids), and to other interdiction problems more broadly.

\section*{Acknowledgements}
This research was enabled in part by computational resources provided by Calcul Qu\'ebec (\url{https://calculquebec.ca}) and the Digital Research Alliance of Canada (\url{https://alliancecan.ca}).

\bibliographystyle{spbasic}
\bibliography{references}

\begin{thebibliography}{32}
\providecommand{\natexlab}[1]{#1}
\providecommand{\url}[1]{{#1}}
\providecommand{\urlprefix}{URL }
\expandafter\ifx\csname urlstyle\endcsname\relax
  \providecommand{\doi}[1]{DOI~\discretionary{}{}{}#1}\else
  \providecommand{\doi}{DOI~\discretionary{}{}{}\begingroup
  \urlstyle{rm}\Url}\fi
\providecommand{\eprint}[2][]{\url{#2}}

\bibitem[{Amato et~al.(1997)Amato, Cattaneo, and
  Italiano}]{amato1997experimental}
Amato G, Cattaneo G, Italiano GF (1997) Experimental analysis of dynamic
  minimum spanning tree algorithms. In: SODA, Citeseer, vol~97, pp 314--323

\bibitem[{Bader and Burkhardt(2019)}]{bader2019simple}
Bader DA, Burkhardt P (2019) A simple and efficient algorithm for finding
  minimum spanning tree replacement edges. arXiv preprint arXiv:190803473

\bibitem[{Bazgan et~al.(2012)Bazgan, Toubaline, and
  Vanderpooten}]{bazgan2012efficient}
Bazgan C, Toubaline S, Vanderpooten D (2012) Efficient determination of the k
  most vital edges for the minimum spanning tree problem. Computers \&
  Operations Research 39(11):2888--2898

\bibitem[{Bazgan et~al.(2013)Bazgan, Toubaline, and
  Vanderpooten}]{bazgan2013critical}
Bazgan C, Toubaline S, Vanderpooten D (2013) Critical edges/nodes for the
  minimum spanning tree problem: complexity and approximation. Journal of
  Combinatorial Optimization 26:178--189

\bibitem[{Cattaneo et~al.(2010)Cattaneo, Faruolo, Petrillo, and
  Italiano}]{cattaneo2010maintaining}
Cattaneo G, Faruolo P, Petrillo UF, Italiano GF (2010) Maintaining dynamic
  minimum spanning trees: An experimental study. Discrete Applied Mathematics
  158(5):404--425

\bibitem[{Chestnut and Zenklusen(2017)}]{chestnut2017interdicting}
Chestnut SR, Zenklusen R (2017) Interdicting structured combinatorial
  optimization problems with $ \{$0,1$\}$-objectives. Mathematics of Operations
  Research 42(1):144--166

\bibitem[{Cook et~al.(1998)Cook, Cunningham, Pulleyblank, and
  Schrijver}]{cunningham1998combinatorial}
Cook WJ, Cunningham WH, Pulleyblank WR, Schrijver A (1998) Combinatorial
  optimization. John Wiley and Sons, New York

\bibitem[{Frederickson and Solis-Oba(1998)}]{frederickson1998algorithms}
Frederickson GN, Solis-Oba R (1998) Algorithms for measuring perturbability in
  matroid optimization. Combinatorica 18:503--518

\bibitem[{Frederickson and Solis-Oba(1999)}]{frederickson1999increasing}
Frederickson GN, Solis-Oba R (1999) Increasing the weight of minimum spanning
  trees. Journal of Algorithms 33(2):244--266

\bibitem[{Hausbrandt et~al.(2024)Hausbrandt, Bachtler, Ruzika, and Sch{
  \"a}fer}]{hausbrandt2024parametric}
Hausbrandt N, Bachtler O, Ruzika S, Sch{ \"a}fer LE (2024) Parametric matroid
  interdiction. Discrete Optimization 51:100823

\bibitem[{Holm et~al.(2001)Holm, De~Lichtenberg, and Thorup}]{holm2001poly}
Holm J, De~Lichtenberg K, Thorup M (2001) Poly-logarithmic deterministic
  fully-dynamic algorithms for connectivity, minimum spanning tree, 2-edge, and
  biconnectivity. Journal of the ACM (JACM) 48(4):723--760

\bibitem[{Hsu et~al.(1991)Hsu, Jan, Lee, Hung, Chern et~al.}]{hsu1991finding}
Hsu LH, Jan RH, Lee YC, Hung CN, Chern MS, et~al. (1991) Finding the most vital
  edge with respect to minimum spanning tree in weighted graphs. Information
  Processing Letters 39(5):277--281

\bibitem[{Israeli and Wood(2002)}]{israeli2002shortest}
Israeli E, Wood RK (2002) Shortest-path network interdiction. Networks: An
  International Journal 40(2):97--111

\bibitem[{Joret and Vetta(2015)}]{joret2015reducing}
Joret G, Vetta A (2015) Reducing the rank of a matroid. Discrete Mathematics
  and Theoretical Computer Science 17(2):143--156

\bibitem[{Ketkov and Prokopyev(2024)}]{ketkov2024class}
Ketkov SS, Prokopyev OA (2024) On a class of interdiction problems with
  partition matroids: complexity and polynomial-time algorithms. arXiv preprint
  arXiv:240112010

\bibitem[{Kleinert et~al.(2021)Kleinert, Labb{\'e}, Ljubi{\'c}, and
  Schmidt}]{kleinert2021survey}
Kleinert T, Labb{\'e} M, Ljubi{\'c} I, Schmidt M (2021) A survey on
  mixed-integer programming techniques in bilevel optimization. EURO Journal on
  Computational Optimization 9:100007

\bibitem[{Liang(2001)}]{liang2001finding}
Liang W (2001) Finding the k most vital edges with respect to minimum spanning
  trees for fixed k. Discrete Applied Mathematics 113(2-3):319--327

\bibitem[{Libura(1991)}]{libura1991sensitivity}
Libura M (1991) Sensitivity analysis for minimum weight base of a matroid.
  Control and Cybernetics 20(3):7--24

\bibitem[{Lin and Chern(1993)}]{lin1993most}
Lin KC, Chern MS (1993) The most vital edges in the minimum spanning tree
  problem. Information Processing Letters 45(1):25--31

\bibitem[{Linhares and Swamy(2017)}]{linhares2017improved}
Linhares A, Swamy C (2017) Improved algorithms for {MST} and metric-{TSP}
  interdiction. arXiv preprint arXiv:170600034

\bibitem[{Pisinger and Toth(1998)}]{pisinger1998knapsack}
Pisinger D, Toth P (1998) Knapsack problems. Handbook of Combinatorial
  Optimization: Volume 1--3 pp 299--428

\bibitem[{Ribeiro and Toso(2007)}]{ribeiro2007experimental}
Ribeiro CC, Toso RF (2007) Experimental analysis of algorithms for updating
  minimum spanning trees on graphs subject to changes on edge weights. In:
  Experimental Algorithms: 6th International Workshop, WEA 2007, Rome, Italy,
  June 6-8, 2007. Proceedings 6, Springer, pp 393--405

\bibitem[{Salazar-Zendeja(2022)}]{salazar2022models}
Salazar-Zendeja L (2022) Models and algorithms for the minimum spanning tree
  interdiction problem. PhD thesis, Centrale Lille Institut

\bibitem[{Schrijver(2002)}]{schrijver2002history}
Schrijver A (2002) On the history of the transportation and maximum flow
  problems. Mathematical programming 91:437--445

\bibitem[{Shen(1999)}]{shen1999finding}
Shen H (1999) Finding the k most vital edges with respect to minimum spanning
  tree. Acta Informatica 36:405--424

\bibitem[{Smith and Song(2020)}]{smith2020survey}
Smith JC, Song Y (2020) A survey of network interdiction models and algorithms.
  European Journal of Operational Research 283(3):797--811

\bibitem[{Vardy(1997)}]{vardy1997intractability}
Vardy A (1997) The intractability of computing the minimum distance of a code.
  IEEE Transactions on Information Theory 43(6):1757--1766

\bibitem[{Von~Stackelberg(1952)}]{von1952theory}
Von~Stackelberg H (1952) The theory of the market economy. Oxford University
  Press, England

\bibitem[{Wei and Walteros(2022)}]{wei2022integer}
Wei N, Walteros JL (2022) Integer programming methods for solving binary
  interdiction games. European Journal of Operational Research 302(2):456--469

\bibitem[{Wei et~al.(2021)Wei, Walteros, and Pajouh}]{wei2021integer}
Wei N, Walteros JL, Pajouh FM (2021) Integer programming formulations for
  minimum spanning tree interdiction. INFORMS Journal on Computing
  33(4):1461--1480

\bibitem[{Weninger and Fukasawa(2023)}]{weninger2023fast}
Weninger N, Fukasawa R (2023) A fast combinatorial algorithm for the bilevel
  knapsack problem with interdiction constraints. In: International Conference
  on Integer Programming and Combinatorial Optimization, Springer, pp 438--452

\bibitem[{Zenklusen(2015)}]{zenklusen20151}
Zenklusen R (2015) An {O(1)}-approximation for minimum spanning tree
  interdiction. In: IEEE 56th Annual Symposium on Foundations of Computer
  Science, IEEE, pp 709--728

\end{thebibliography}

\appendix
\section{Proof of \texorpdfstring{\cref{lem:directsum}}{Theorem \ref{lem:directsum}}}
\label{apx:directsum}
\begin{proof}
First we show that $\propbound(\bound,\state,\update)$ is satisfied. Let $r\in\{0,\dots,C\}$, $s\in T$ and $i\in E$ such that
$c_i\le r$.
We claim that for all $X\in\U$,
\begin{align}
i\le p\implies F_{M_1}(X_{<i}\cup\{i\})-F_{M_1}(X_{<i})=F_M(X_{<i}\cup\{i\})-F_M(X_{<i}),\label{eqn:dsclaim1}\\
i> p\implies F_{M_2}(X_{<i}\cup\{i\})-F_{M_2}(X_{<i})=F_M(X_{<i}\cup\{i\})-F_M(X_{<i}).\label{eqn:dsclaim2}
\end{align}
To see this, observe that for any $X'\se E$, $F_M(X')=F_{M_1}(X')+F_{M_2}(X')$
because the independent sets of $M_1$ and $M_2$ are disjoint. Now, considering what elements
can be independent in each matroid, if $i\le p$, then
$F_{M_2}(X_{<i}\cup\{i\})=F_{M_2}(\emptyset)=F_{M_2}(X_{<i})$ and if $i>p$ then
$F_{M_1}(X_{<i}\cup\{i\})=F_{M_1}(X\cap\{1,\dots,p\})=F_{M_1}(X_{<i})$, establishing the claim.

Now we establish $\propbound(\bound,\state,\update)$.
If $i\le p$, then
\begin{align*}
\bound(i,r,s)&=\bound_1(i,r,s)\\
	&\ge\max\{F_M(X_{<i}\cup\{i\})-F_M(X_{<i}):X\in\U,\,C-c(X_{<i})=r,\,\state_1(X_{<i})=s\}
		\tag{by $\propbound(\bound_1,\state_1,\update_1)$ and \cref{eqn:dsclaim1}}\\
	&=\max\{F_M(X_{<i}\cup\{i\})-F_M(X_{<i}):X\in\U,\,C-c(X_{<i})=r,\,\state(X_{<i})=s\}.
		\tag{because $i\le p$ implies $\state(X_{<i})=\state_1(X_{<i})$}
\end{align*}
So, assume $i>p$. If $s\in T_2$ then
\begin{align*}
\bound(i,r,s)&=\bound_2(i,r,s)\\
	&\ge\max\{F_M(X_{<i}\cup\{i\})-F_M(X_{<i}):X\in\U,\,C-c(X_{<i})=r,\,\state_2(X_{<i})=s\}
		\tag{by $\propbound(\bound_2,\state_2,\update_2)$ and \cref{eqn:dsclaim2}}\\
	&\ge\max\{F_M(X_{<i}\cup\{i\})-F_M(X_{<i}):X\in\U,\,C-c(X_{<i})=r,\,\state(X_{<i})=s\}
\end{align*}
where the last inequality follows because if $X\in\U$ satisfies $\state(X_{<i})=s$,
then since $s\in T_2$, we know $\state(X_{<i})=\state_2(X_{<i})$.
Finally, if $i>p$ and $s\in T_1$, then
\begin{align*}
\bound(i,r,s)&=\bound_2(i,r,\state_2(\emptyset))\\
	&\ge\max\{F_M(X_{<i}\cup\{i\})-F_M(X_{<i}):X\in\U,\,C-c(X_{<i})=r,\,\state_2(X_{<i})=\state_2(\emptyset)\}
		\tag{by $\propbound(\bound_2,\state_2,\update_2)$ and \cref{eqn:dsclaim2}}\\
	&\ge\max\{F_M(X_{<i}\cup\{i\})-F_M(X_{<i}):X\in\U,\,C-c(X_{<i})=r,\,\state(X_{<i})=s\}
\end{align*}
where the last inequality follows because if $X\in\U$ satisfies $\state(X_{<i})=s$,
then since $s\in T_1$, we have that $X\se\{1,\dots,p\}$. Hence, $\state_2(\emptyset)=\state_2(X_{<i})$
by \cref{eqn:dsassume}.

Now we show that $\propiinx(\state,\update)$ holds.
We consider four cases:
\begin{enumerate}
\item
If $X_{<i}\se\{1,\dots,p\}$ and $i\le p$: since $\propiinx(\state_1,\update_1)$ holds,
\[\state(X_{<i}\cup\{i\})=\state_1(X_{<i}\cup\{i\})=\update_1(i,\state_1(X_{<i}))=\update(i,\state(X_{<i})).\]

\item
If $X_{<i}\se\{1,\dots,p\}$ and $i>p$: by $\propiinx(\state_2,\update_2)$ and \cref{eqn:dsassume},
\[\state(X_{<i}\cup\{i\})=\state_2(X_{<i}\cup\{i\})=\update_2(i,\state_2(X_{<i}))=\update_2(i,\state_2(\emptyset))=\update(i,\state(X_{<i})).\]

\item
If $X_{<i}\not\se\{1,\dots,p\}$ and $i\le p$: this is impossible, as $X_{<i}\not\se\{1,\dots,p\}$ implies $i>p$.

\item
If $X_{<i}\not\se\{1,\dots,p\}$ and $i>p$: since $\propiinx(\state_2,\update_2)$ holds,
\[\state(X_{<i}\cup\{i\})=\state_2(X_{<i}\cup\{i\})=\update_2(i,\state_2(X_{<i}))=\update(i,\state(X_{<i})).\]
\end{enumerate}

Finally, we show that if there exists some $\alpha\ge0$ such that
$\propapprox(\bound_j,\state_j,\update_j,\alpha)$ is satisfied for both $j=1$ and $j=2$, then
$\propapprox(\bound,\state,\update,\alpha)$ is satisfied.
Let $X\se\U$, $i\in E$.
If $i\le p$, then by $\propapprox(\bound_1,\state_1,\update_1,\alpha)$ and \cref{eqn:dsclaim1},
\begin{align*}
\bound(i,C-c(X_{<i}),\state(X_{<i}))&=\bound_1(i,C-c(X_{<i}),\state(X_{<i}))\\
	&\le\alpha+F_M(X_{<i}\cup\{i\})-F_M(X_{<i}).
\end{align*}
So, assume $i>p$. Then, if $X\not\se\{1,\dots,p\}$, then by $\propapprox(\bound_2,\state_2,\update_2,\alpha)$ and \cref{eqn:dsclaim2},
\begin{align*}
\bound(i,C-c(X_{<i}),\state(X_{<i}))&=\bound_2(i,C-c(X_{<i}),\state(X_{<i}))\\
	&\le\alpha+F_M(X_{<i}\cup\{i\})-F_M(X_{<i}).
\end{align*}
Otherwise, if $X\se\{1,\dots,p\}$, 
then by \cref{eqn:dsassume}, $\propapprox(\bound_2,\state_2,\update_2,\alpha)$ and \cref{eqn:dsclaim2},
\begin{align*}
\bound(i,C-c(X_{<i}),\state(X_{<i}))&=\bound_2(i,C-c(X_{<i}),\state_2(\emptyset))\\
	&=\bound_2(i,C-c(X_{<i}),\state_2(X_{<i}))\\
	&\le\alpha+F_M(X_{<i}\cup\{i\})-F_M(X_{<i}).\tag*{\qedhere}
\end{align*}
\end{proof}

\end{document}